\documentclass[a4paper,UKenglish,cleveref,thm-restate]{lipics-v2019}

\title{HyperLTL Satisfiability is $\Sigma_1^1$-complete, HyperCTL* Satisfiability is $\Sigma_1^2$-complete}

\author{Marie Fortin}{University of Liverpool, UK}{Marie.Fortin@liverpool.ac.uk}{https://orcid.org/0000-0001-5278-0430}{}
\author{Louwe B. Kuijer}{University of Liverpool, UK}{Louwe.Kuijer@liverpool.ac.uk}{https://orcid.org/0000-0001-6696-9023}{}
\author{Patrick Totzke}{University of Liverpool, UK}{totzke@liverpool.ac.uk}{http://orcid.org/0000-0001-5274-8190}{}
\author{Martin Zimmermann}{University of Liverpool, UK}{Martin.Zimmermann@liverpool.ac.uk}{http://orcid.org/0000-0002-8038-2453}{}

\authorrunning{M.~Fortin and L.~Kuijer and P.~Totzke and M.~Zimmermann}

\Copyright{M.~Fortin and L.~Kuijer and P.~Totzke and M.~Zimmermann}

\ccsdesc[500]{Theory of computation~Logic and verification}
\ccsdesc[500]{Theory of computation~Formal languages and automata theory}

\keywords{HyperLTL, HyperCTL*, Satisfiability, Analytical Hierarchy}

\category{} %

\relatedversion{} %

\supplement{}%

\funding{Partially funded by EPSRC grants EP/S032207/1 and EP/V025848/1.}%

\nolinenumbers
\hideLIPIcs

\EventEditors{John Q. Open and Joan R. Access}
\EventNoEds{2}
\EventLongTitle{42nd Conference on Very Important Topics (CVIT 2016)}
\EventShortTitle{CVIT 2016}
\EventAcronym{CVIT}
\EventYear{2016}
\EventDate{December 24--27, 2016}
\EventLocation{Little Whinging, United Kingdom}
\EventLogo{}
\SeriesVolume{42}
\ArticleNo{23}

\usepackage[utf8]{inputenc}

\usepackage{amsmath}
\usepackage{amssymb}
\usepackage{lettrine}

\usepackage{xspace}
\usepackage{turnstile}

\usepackage{booktabs}
\usepackage{multirow}
\usepackage{pifont}

\usepackage{multicol}
\usepackage{stackengine}
\usepackage{pgf}

\usepackage{tikz}
\usetikzlibrary{decorations.pathreplacing,positioning}
\usetikzlibrary{arrows,automata}
\tikzset{>=stealth, shorten >=1pt}
\tikzset{every edge/.style = {thick, ->, draw}}
\tikzset{every loop/.style = {thick, ->, draw}}
\pgfdeclarelayer{background}
\pgfsetlayers{background,main}

\newcommand{\expansion}{e}

\newcommand{\contcard}{\cont}
\newcommand{\mcgame}{\mathcal{G}}

\newcommand{\last}{\mathrm{rcnt}}

\newcommand{\set}[1]{\{#1\}}
\newcommand{\nats}{\mathbb{N}}

\renewcommand{\epsilon}{\varepsilon}

\renewcommand{\phi}{\varphi}
\newcommand{\size}[1]{|#1|}

\newcommand{\N}{\mathbb{N}}
\newcommand\Nat{\mathbb{N}}

\newcommand{\x}{\times}
\newcommand{\pow}[1]{2^{#1}}
\newcommand{\cont}{\mathfrak{c}}

\newcommand{\cceq}{\mathop{::=}}

\newcommand{\tsys}{\mathcal{T}}

\newcommand{\T}{\mathcal{T}}

\newcommand{\initmark}{I}

\newcommand{\F}{\mathop{\mathbf{F}\vphantom{a}}\nolimits}
\newcommand{\G}{\mathop{\mathbf{G}\vphantom{a}}\nolimits}
\DeclareMathOperator{\U}{\mathbf{U}}
\newcommand{\X}{\mathop{\mathbf{X}\vphantom{a}}\nolimits}

\newcommand{\ltl}{{LTL}\xspace}
\newcommand{\ctlstar}{{CTL$^*$}\xspace}
\newcommand{\hyltl}{{HyperLTL}\xspace}
\newcommand{\hyctlstar}{{HyperCTL$^*$}\xspace}
\newcommand{\fol}{{FO}$[<]$\xspace}

\newcommand\FO{\textup{FO}}

\newcommand\HyperLTL{\textup{HyperLTL}}
\newcommand\simpleHyperLTL{\encPhi {\FOw}}

\newcommand{\suffix}[2]{#1[#2,\infty)}
\newcommand{\var}{\mathcal{V}}
\newcommand{\ap}[0]{\mathrm{AP}}
\newcommand\AP{\textup{AP}}
\newcommand\Free{\textup{Free}}
\newcommand\depth{\textit{depth}}

\newcommand\FOw{\FO[\le]}
\newcommand\varfo{\textit{Var}}
\newcommand\pos{\textit{pos}}
\newcommand\Sfo[1]{\Sigma_{#1}(\FOw)}

\newcommand\enc[2]{\mathit{enc}(#1,#2)}
\newcommand\encn[2]{\mathit{enc}_{#2}(#1)}
\newcommand\Tn[2]{#1^{(#2)}}
\newcommand\nun[2]{#1^{(#2)}}
\newcommand\encPhi[1]{\mathit{enc}(#1)}

\newcommand\simpl[1]{\widehat {#1}}

\newcommand\Sign[1]{\Sigma_{#1}}
\newcommand\Pin[1]{\Pi_{#1}}
\newcommand\Tl{T_\ell}
\newcommand\Tr{T_r}

\newcommand\phib{\varphi_{\mathit{bd}}}
\newcommand\Vl{V_\ell}
\newcommand\Vr{V_r}
\newcommand\Sigmal{\Sigma_\ell}
\newcommand\Sigmar{\Sigma_r}

\newcommand\argl{\textup{\texttt{arg1}}}
\newcommand\argr{\textup{\texttt{arg2}}}
\newcommand\res{\textup{\texttt{res}}}
\newcommand\add{\textup{\texttt{add}}}
\newcommand\mult{\textup{\texttt{mult}}}
\newcommand\fbt{\textup{\texttt{fbt}}}
\newcommand\pset{\textup{\texttt{set}}}
\newcommand\phiop{\varphi_{{+},{\cdot}}}
\newcommand\Top{T_{{+},{\cdot}}}
\newcommand\phiset{\varphi_{\cont}}
\newcommand\Kset{\tsys_{\cont}}
\newcommand\xo{x}
\newcommand\yo{y}
\newcommand\zo{z}

\newcommand\ys{y}
\newcommand\xt{x}
\newcommand\yt{y}
\newcommand\ai[1]{a_{#1}}
\newcommand\ax{a}
\newcommand\pix[1]{\pi_{#1}}

\definecolor{com}{RGB}{0,128,0}

\renewcommand{\phi}{\varphi}
\newcommand{\east}{\mathit{east}}
\newcommand{\west}{\mathit{west}}
\newcommand{\north}{\mathit{north}}
\newcommand{\south}{\mathit{south}}

\begin{document}

\maketitle

\begin{abstract}
Temporal logics for the specification of information-flow properties are able to express relations between multiple executions of a system.
The two most important such logics are HyperLTL and HyperCTL*, which generalise LTL and CTL* by trace quantification. 
It is known that this expressiveness comes at a price, i.e.\ satisfiability is undecidable for both logics. 

In this paper we settle the exact complexity of these problems, showing that both are in fact highly undecidable:
we prove that HyperLTL satisfiability is $\Sigma_1^1$-complete and HyperCTL* satisfiability is $\Sigma_1^2$-complete. 
These are significant increases over the previously known lower bounds and the first upper bounds.
To prove $\Sigma_1^2$-membership for HyperCTL*, we prove that every satisfiable HyperCTL* sentence has a model that is equinumerous to the continuum, the first upper bound of this kind. We prove this bound to be tight.
Finally, we show that the membership problem for every level of the HyperLTL quantifier alternation hierarchy is $\Pi_1^1$-complete.

\end{abstract}

\section{Introduction}
\label{sec:intro}
Most classical temporal logics like \ltl and \ctlstar refer to a single execution trace at a time while information-flow properties, which are crucial for security-critical systems,
require reasoning about multiple executions of a system. 
Clarkson and Schneider~\cite{ClarksonS10} coined the term \emph{hyperproperties} for such properties which, structurally, are sets \emph{of sets} of traces.
Just like ordinary trace and branching-time properties, hyperproperties can be specified using temporal logics, e.g.\ \hyltl and \hyctlstar~\cite{ClarksonFKMRS14}, expressive, but intuitive specification languages that are able to express typical information-flow properties such as noninterference, noninference, declassification, and input determinism.
Due to their practical relevance and theoretical elegance, hyperproperties and their specification languages have received considerable attention during the last decade~\cite{DBLP:conf/atva/AbrahamBBD20,DBLP:conf/qest/AbrahamB18,bartocci2021flavours,abs-2104-14025,BerardHH18,bozzelli2021asynchronous,ClarksonFKMRS14,ClarksonS10,CoenenFHH19,DimitrovaFT20,FZ17,GutsfeldMO20,HO2020104639,DBLP:conf/mfcs/KrebsMV018,DBLP:journals/corr/abs-2010-03311}.

\hyltl is obtained by extending \ltl~\cite{Pnueli77}, the most influential specification language for linear-time properties, by trace quantifiers to refer to multiple executions of a system. 
For example, the \hyltl formula
\[\forall \pi.\ \forall \pi'.\ \G( i_\pi \leftrightarrow i_{\pi'}) \rightarrow \G (o_\pi \leftrightarrow o_{\pi'}) \]
expresses input determinism, i.e.\ every pair of traces that always has the same input (represented by the proposition~$i$) also always has the same output (represented by the proposition~$o$). 
Similarly, \hyctlstar is the extension of the branching-time logic \ctlstar~\cite{EmersonH86} by path quantifiers. 
\hyltl only allows formulas in prenex normal form while \hyctlstar allows arbitrary quantification, in particular under the scope of temporal operators. 
Consequently, \hyltl formulas are evaluated over sets of traces while \hyctlstar formulas are evaluated over transition systems, which yield the underlying branching structure of the traces. 

All basic verification problems, e.g.\ model checking~\cite{Finkbeiner21,FinkbeinerRS15}, runtime monitoring~\cite{AgrawalB16,BonakdarpourF16,BrettSB17,FinkbeinerHST18}, and synthesis~\cite{BonakdarpourF20,FinkbeinerHHT20,FinkbeinerHLST20}, have been studied. Most importantly, \hyctlstar model checking over finite transition systems is decidable and TOWER-complete for a fixed transition system~\cite{FinkbeinerRS15, MZ20}. 
However, for a small number of alternations, efficient algorithms have been developed and were applied to a wide range of problems, e.g.\ an information-flow analysis of an I2C bus master~\cite{FinkbeinerRS15}, the symmetric access to a shared resource in a mutual exclusion protocol~\cite{FinkbeinerRS15}, and to detect the use of a defeat device to cheat in emission testing~\cite{BartheDFH16}.

But surprisingly, the exact complexity of the satisfiability problems for \hyltl and \hyctlstar is still open. 
Finkbeiner and Hahn proved that \hyltl satisfiability is undecidable~\cite{FinkbeinerH16}, a result which already holds when only considering finite sets of ultimately periodic traces
 and $\forall\exists$-formulas. 
In fact, Finkbeiner et al.\ showed that \hyltl satisfiability restricted to finite sets of ultimately periodic traces is $\Sigma_1^0$-complete~\cite{FHH18} (i.e.\ complete for the set of recursively enumerable problems).
Furthermore, Hahn and Finkbeiner proved that the $\exists^*\forall^*$-fragment has decidable satisfiability~\cite{FinkbeinerH16} while Mascle and Zimmermann studied the \hyltl satisfiability problem restricted to bounded sets of traces~\cite{MZ20}. The latter work implies that \hyltl satisfiability restricted to finite sets of traces (even non ultimately periodic ones) is also $\Sigma_1^0$-complete.
Finally, Finkbeiner et al.\ developed tools and heuristics~\cite{FHH18,FinkbeinerHS17}.

As every \hyltl formula can be turned into an equisatisfiable \hyctlstar formula, \hyctlstar satisfiability is also undecidable. 
Moreover, Rabe has shown that it is even $\Sigma_1^1$-hard~\cite{Rabe16}, i.e.\ it is not even arithmetical.
However, both for \hyltl and for \hyctlstar satisfiability, only lower bounds, but no upper bounds, are known. 

\subparagraph*{Our Contributions.}
In this paper, we settle the complexity of the satisfiability problems for \hyltl and \hyctlstar
by determining exactly how undecidable they are.
That is, we provide matching lower and upper bounds in terms of the analytical hierarchy and beyond,
where decision problems (encoded as subsets of $\nats$) are classified based on their definability by formulas of higher-order arithmetic, namely by the type of objects one can quantify over and by the number of alternations of such quantifiers.
We refer to Roger's textbook~\cite{Rogers87} for fully formal definitions.
For our purposes, it suffices to recall the following classes.
$\Sigma_1^0$ contains the sets of natural numbers of the form
\[
\set{x \in \nats \mid \exists x_0 \cdots \exists x_k.\ \psi(x, x_0, \ldots, x_k)}
\] 
where quantifiers range over natural numbers and $\psi$ is a quantifier-free arithmetic formula.
The notation~$\Sigma_1^0$ signifies that there is a single block of existential quantifiers (the subscript~$1$) ranging over natural numbers (type~$0$ objects, explaining the superscript~$0$).
Analogously, $\Sigma_1^1$ is induced by arithmetic formulas with existential quantification of type~$1$ objects (functions mapping natural numbers to natural numbers) and arbitrary (universal and existential) quantification of type~$0$ objects.
Finally, $\Sigma_1^2$ is induced by arithmetic formulas with existential quantification of type~$2$ objects (functions mapping type~$1$ objects to natural numbers) and arbitrary  quantification of type~$0$ and type~$1$ objects.
So, $\Sigma_1^0$ is part of the first level of the arithmetic hierarchy, $\Sigma_1^1$ is part of the first level of the analytical hierarchy, while $\Sigma_1^2$ is not even analytical.

In terms of this classification, we prove that \hyltl satisfiability is $\Sigma_1^1$-complete while \hyctlstar satisfiability is $\Sigma_1^2$-complete, thereby settling the complexity of both problems
and showing that they are highly undecidable.
In both cases, this is a significant increase of the lower bound and the first upper bound.

First, let us consider \hyltl satisfiability.
The $\Sigma_1^1$ lower bound is a reduction from the recurrent tiling problem, a standard $\Sigma_1^1$-complete problem asking whether $\nats\times\nats$ can be tiled by a given finite set of tiles.
So, let us consider the upper bound: $\Sigma_1^1$ allows to quantify over type~$1$ objects: functions from natural numbers to natural numbers, or, equivalently, over sets of natural numbers, i.e.\ countable objects.
On the other hand, \hyltl formulas are evaluated over sets of infinite traces, i.e.\ uncountable objects. 
Thus, to show that quantification over type~$1$ objects is sufficient, we need to apply a result of Finkbeiner and Zimmermann proving that every satisfiable \hyltl formula has a countable model~\cite{FZ17}.
Then, we can prove $\Sigma_1^1$-membership by expressing the existence of a model and the existence of appropriate Skolem functions for the trace quantifiers by type~$1$ quantification. 
We also prove that the satisfiability problem remains $\Sigma_1^1$-complete when restricted to ultimately periodic traces, or, equivalently, when restricted to finite traces.

Then, we turn our attention to \hyctlstar satisfiability.
Recall that \hyctlstar formulas are evaluated over (possibly infinite) transition systems, which can be much larger than type~$2$ objects whose cardinality is bounded by $\contcard$, the cardinality of the continuum.
Hence, to obtain our upper bound on the complexity we need, just like in the case of \hyltl, an upper bound on the size of minimal models of satisfiable \hyctlstar formulas.
To this end, we generalise the proof of Finkbeiner and Zimmermann to \hyctlstar, showing that every satisfiable \hyctlstar formula has a model of size~$\contcard$. 
We also exhibit a satisfiable \hyctlstar formula~$\varphi_\contcard$ whose models all have at least cardinality~$\contcard$, as they have to encode all subsets of $\nats$ by disjoint paths.
Thus, our upper bound~$\contcard$ is tight.

With this upper bound on the cardinality of models, we are able to prove $\Sigma_1^2$-membership of \hyctlstar satisfiability by expressing with type~$2$ quantification the existence of a model and the existence of a winning strategy in the induced model checking game.
The matching lower bound is proven by directly encoding the arithmetic formulas inducing $\Sigma_1^2$ as instances of the \hyctlstar satisfiability problem.
To this end, we use the formula~$\varphi_\contcard$ whose models have for each subset $A \subseteq \nats$ a path encoding $A$.
Now, quantification over type~$0$ objects (natural numbers) is simulated by quantification of a path encoding a singleton set, quantification over type~$1$ objects (which can be assumed to be sets of natural numbers) is simulated by quantification over the paths encoding such subsets, and existential quantification over type~$2$ objects (which can be assumed to be subsets of $\pow{\nats}$) is simulated by the choice of the model, i.e.\ a model encodes $k$ subsets of $\pow{\nats}$ if there are $k$ existential type~$2$ quantifiers.
Finally, the arithmetic operations can easily be implemented in \hyltl, and therefore also in \hyctlstar.

After settling the complexity of satisfiability, we turn our attention to the \hyltl quantifier alternation hierarchy and its relation to satisfiability. 
Rabe remarks that the hierarchy is strict~\cite{Rabe16}, and Mascle and Zimmermann show that every \hyltl formula has an equi-satisfiable formula with one quantifier alternation~\cite{MZ20}.
Here, we present a novel proof of strictness by embedding the \fol alternation hierarchy, which is also strict~\cite{CohenB71,Thomas81}. 
We use our construction to prove that deciding equivalence to formulas of fixed quantifier alternation is $\Pi_1^1$-complete (i.e.\ the co-class of $\Sigma_1^1$).

All proofs omitted due to space restrictions can be found in the appendix.

\section{Preliminaries}
\label{sec:definitions}
Fix a finite set~$\ap$ of atomic propositions. A \emph{trace} over $\ap$ is a map $t \colon \nats \rightarrow \pow{\ap}$, denoted by $t(0)t(1)t(2) \cdots$. 
It is \emph{ultimately periodic}, if $t = x \cdot y^\omega$ for some $x,y \in (\pow{\ap})^+$, i.e.\ there are $s,p>0$ with $t(n) = t(n+p)$ for all $n \ge s$.
The set of all traces over $\ap$ is  $(\pow{\ap})^\omega$. 

A transition system~$\tsys = (V,E,v_\initmark, \lambda)$ consists of a set~$V$ of vertices, a set~$E \subseteq V \times V$ of (directed) edges, an initial vertex~$v_\initmark \in V$, and a labelling~$\lambda\colon V \rightarrow \pow{\ap}$ of the vertices by sets of atomic propositions. A path $\rho$ through~$\tsys$ is an infinite sequence~$\rho(0)\rho(1)\rho(2)\cdots$ of vertices with $(\rho(n),\rho(n+1))\in E$ for every $n \ge 0$.

\subparagraph*{\hyltl.}

The formulas of \hyltl are given by the grammar
\[
\phi  {} \cceq {}  \exists \pi.\ \phi \mid \forall \pi.\ \phi \mid \psi \qquad\qquad
\psi {}  \cceq {}  a_\pi \mid \neg \psi \mid \psi \vee \psi \mid \X \psi \mid \psi \U \psi
\]
where $a$ ranges over atomic propositions in $\ap$ and where $\pi$ ranges over a fixed countable set~$\var$ of \emph{(trace) variables}. Conjunction, implication, and equivalence are defined as usual, and the temporal operators eventually~$\F$ and always~$\G$ are derived as $\F\psi = \neg \psi\U \psi$ and $\G \psi = \neg \F \neg \psi$. A \emph{sentence} is a  formula without free variables.

The semantics of \hyltl is defined with respect to a \emph{trace assignment}, a partial mapping~$\Pi \colon \var \rightarrow (\pow{\ap})^\omega$. The assignment with empty domain is denoted by $\Pi_\emptyset$. Given a trace assignment~$\Pi$, a variable~$\pi$, and a trace~$t$ we denote by $\Pi[\pi \rightarrow t]$ the assignment that coincides with $\Pi$ everywhere but at $\pi$, which is mapped to $t$. 
Furthermore, $\suffix{\Pi}{j}$ denotes the trace assignment mapping every $\pi$ in $\Pi$'s domain to $\Pi(\pi)(j)\Pi(\pi)(j+1)\Pi(\pi)(j+2) \cdots $.

For sets~$T$ of traces and trace assignments~$\Pi$ we define 
\begin{itemize}
	\item $(T, \Pi) \models a_\pi$ if $a \in \Pi(\pi)(0)$,
	\item $(T, \Pi) \models \neg \psi$ if $(T, \Pi) \not\models \psi$,
	\item $(T, \Pi) \models \psi_1 \vee \psi_2 $ if $(T, \Pi) \models \psi_1$ or $(T, \Pi) \models \psi_2$,
	\item $(T, \Pi) \models \X \psi$ if $(T,\suffix{\Pi}{1}) \models \psi$,
	\item $(T, \Pi) \models \psi_1 \U \psi_2$ if there is a $j \ge 0$ such that $(T,\suffix{\Pi}{j}) \models \psi_2$ and for all $0 \le j' < j$: $(T,\suffix{\Pi}{j'}) \models \psi_1$, 
	\item $(T, \Pi) \models \exists \pi.\ \phi$ if there exists a trace~$t \in T$ such that $(T,\Pi[\pi \rightarrow t]) \models \phi$, and 
	\item $(T, \Pi) \models \forall \pi.\ \phi$ if for all traces~$t \in T$: $(T,\Pi[\pi \rightarrow t]) \models \phi$. 
\end{itemize}
We say that $T$ \emph{satisfies} a sentence~$\phi$ if $(T, \Pi_\emptyset) \models \phi$. In this case, we write $T \models \phi$ and say that $T$ is a \emph{model} of $\phi$. 
Although \hyltl sentences are required to be in prenex normal form, they are closed under Boolean combinations, which can easily be seen by transforming such formulas into prenex normal form. Two \hyltl sentences~$\varphi$ and $\varphi'$ are equivalent if $T \models \varphi$ if and only if $T \models \varphi'$ for every set~$T$ of traces.

\subparagraph*{\hyctlstar.}

The formulas of \hyctlstar are given by the grammar
\begin{align*}
\phi & {} \cceq {} a_\pi \mid \neg \phi \mid \phi \vee \phi \mid \X \phi \mid \phi \U \phi \mid \exists \pi.\ \phi \mid \forall \pi.\ \phi
\end{align*}
where $a$ ranges over atomic propositions in $\ap$ and where $\pi$ ranges over a fixed countable set~$\var$ of \emph{(path) variables}, and where we require that each temporal operator appears in the scope of a path quantifier. Again, other Boolean connectives and temporal operators are derived as usual.
Sentences are formulas without free variables.

Let $\tsys$ be a transition system. The semantics of \hyctlstar is defined with respect to a \emph{path assignment}, a partial mapping~$\Pi$ from variables in $\var$ to paths of $\tsys$. The assignment with empty domain is denoted by $\Pi_\emptyset$. Given a path assignment~$\Pi$, a variable~$\pi$, and a path~$\rho$ we denote by $\Pi[\pi \rightarrow \rho]$ the assignment that coincides with $\Pi$ everywhere but at $\pi$, which is mapped to $\rho$. 
Furthermore, $\suffix{\Pi}{j}$ denotes the trace assignment mapping every $\pi$ in $\Pi$'s domain to $\Pi(\pi)(j)\Pi(\pi)(j+1)\Pi(\pi)(j+2) \cdots $, its suffix from position $j$ onwards.
 For transition systems~$\tsys$ and paths assignments~$\Pi$ we define 
\begin{itemize}
	\item $(\tsys, \Pi) \models a_\pi$ if $a \in \lambda(\Pi(\pi)(0))$, where $\lambda$ is the labelling function of $\tsys$,
	\item $(\tsys, \Pi) \models \neg \psi$ if $(\tsys, \Pi) \not\models \psi$,
	\item $(\tsys, \Pi) \models \psi_1 \vee \psi_2 $ if $(\tsys, \Pi) \models \psi_1$ or $(\tsys, \Pi) \models \psi_2$,
	\item $(\tsys, \Pi) \models \X \psi$ if $(\tsys,\suffix{\Pi}{1}) \models \psi$,
	\item $(\tsys, \Pi) \models \psi_1 \U \psi_2$ if there exists a $j \ge 0$ such that $(\tsys,\suffix{\Pi}{j}) \models \psi_2$ and for all $0 \le j' < j$: $(\tsys,\suffix{\Pi}{j'}) \models \psi_1$, 
	\item $(\tsys, \Pi) \models \exists \pi.\ \phi$ if there exists a path~$\rho$ of $\tsys$, starting in $\last(\Pi)$, such that $(\tsys,\Pi[\pi \rightarrow \rho]) \models \phi$, and 
	\item $(\tsys, \Pi) \models \forall \pi.\ \phi$ if for all paths~$\rho$ of $\tsys$ starting in $\last(\Pi)$: $(\tsys,\Pi[\pi \rightarrow \rho]) \models \phi$. 
\end{itemize}
Here, $\last(\Pi)$ is the initial vertex of $\Pi(\pi)$, where $\pi$ is the path variable most recently added to $\Pi$, and the initial vertex of $\tsys$ if $\Pi$ is empty.\footnote{For the sake of simplicity, we refrain from formalising this notion properly, which would require to keep track of the order in which variables are added to $\Pi$'s domain.}
We say that $\tsys$ \emph{satisfies} a sentence~$\phi$ if $(\tsys, \Pi_\emptyset) \models \phi$. In this case, we write $\tsys \models \phi$ and say that $\tsys$ is a \emph{model} of $\phi$. 

\subparagraph*{Complexity Classes for Undecidable Problems.}
A type~$0$ object is a natural number~$n \in \nats$, a type~$1$ object is a function $f \colon \nats \rightarrow \nats$, and a type~$2$ object is a function~$f\colon (\nats \rightarrow \nats) \rightarrow \nats$.
As usual, predicate logic with quantification over type~$0$ objects (first-order quantifiers) is called first-order logic.
Second- and third-order logic are defined similarly.

We consider formulas of arithmetic, i.e.\ predicate logic with signature~$(0,1, +, \cdot, <)$ evaluated over the natural numbers. With a single free variable of type~$0$, such formulas define sets of natural numbers (see, e.g.\ Rogers~\cite{Rogers87} for more details):
\begin{itemize}
	\item $\Sigma_1^0$ contains the sets of the form $\set{x \in \nats \mid \exists x_0 \cdots \exists x_k.\ \psi(x, x_0, \ldots, x_k)}$ where $\psi$ is a quantifier-free arithmetic formula and the $x_i$ are variables of type~$0$. 
	\item $\Sigma_1^1$ contains the sets of the form $\set{x \in \nats \mid \exists x_0 \cdots \exists x_k.\ \psi(x, x_0, \ldots, x_k)}$ where $\psi$ is an arithmetic formula with arbitrary (existential and universal) quantification over type~$0$ objects and the $x_i$ are variables of type~$1$. 
	\item $\Sigma_1^2$ contains the sets of the form $\set{x \in \nats \mid \exists x_0 \cdots \exists x_k.\ \psi(x, x_0, \ldots, x_k)}$ where $\psi$ is an arithmetic formula with arbitrary (existential and universal) quantification over type~$0$ and type~$1$ objects and the $x_i$ are variables of type~$2$.
\end{itemize}
Note that there is a bijection between functions of the form~$f\colon \nats\rightarrow \nats$ and subsets of $\nats$, which is implementable in  arithmetic. 
Similarly, there is a bijection between functions of the form~$f \colon (\nats \rightarrow \nats)\rightarrow\nats$ and subsets of $\pow{\nats}$, which is again implementable in arithmetic.
Thus, whenever convenient, we use quantification over sets of natural numbers and over sets of sets of natural numbers, instead of quantification over type~$1$ and type~$2$ objects; in particular when proving lower bounds. We then include $\in$ in the signature.

\section{HyperLTL satisfiability is $\Sigma_1^1$-complete}
\label{sec:sat}
In this section we settle the complexity of the satisfiability problem for HyperLTL:
given a \hyltl sentence, determine whether it has a model.

\begin{theorem}\label{thm:hyltl-sat}
  \hyltl satisfiability is $\Sigma^1_1$-complete.
\end{theorem}

We should contrast this result with \cite[Theorem 1]{FHH18}, which shows that
\hyltl satisfiability by
\emph{finite} sets of ultimately periodic traces is $\Sigma^0_1$-complete. The $\Sigma^1_1$-completeness
of \hyltl satisfiability in the general case implies that, in particular, the set of
satisfiable \hyltl sentences is neither recursively enumerable nor co-recursively
enumerable. A semi-decision procedure, like the one introduced in \cite{FHH18}
for finite sets of ultimately periodic traces, therefore cannot exist in general.

\begin{proof}[Proof sketch]
The $\Sigma^1_1$ upper bound relies on the fact that every satisfiable \hyltl
formula has a countable model~\cite{FZ17}. 
This allows us to represent these models, and Skolem functions on them, by sets of natural numbers, which are 
type~$1$ objects. In this encoding, trace assignments are type~$0$ objects, as traces in a countable set can be identified by natural numbers. With some more existential type~$1$ quantification one can then express the existence of a function witnessing that every trace assignment consistent with the Skolem functions satisfies the quantifier-free part of the formula under consideration. Hence, \hyltl
satisfiability is in $\Sigma^1_1$.

\smallskip

We show hardness of HyperLTL satisfiability by a reduction from the \emph{recurring tiling problem}
which is given as follows.
A \emph{tile} is a function $\tau\colon\{\east,\west,\north,\south\}\to \mathcal{C}$
that maps directions into a finite set $\mathcal{C}$ of colours.
Given a finite set $\mathcal{T}$ of tiles, a \emph{tiling of the positive quadrant}
with $\mathcal{T}$ is a function $T:\mathbb{N}\times\mathbb{N}\to\mathcal{T}$ with the property that:
\begin{itemize}
	\item if $T(i,j)=\tau_1$ and $T(i+1,j)=\tau_2$, then $\tau_1(\east)=\tau_2(\west)$ and
	\item if $T(i,j)=\tau_1$ and $T(i,j+1)=\tau_2$ then $\tau_1(\north)=\tau_2(\south)$.
\end{itemize}
The \emph{recurring tiling problem} is to determine, given a finite set $\mathcal{T}$ of tiles and a designated $\tau_0\in \mathcal{T}$, whether there is a tiling~$T$ of the positive quadrant with $\mathcal{T}$ such that there are infinitely many $j\in\mathbb{N}$ such that $T(0,j)=\tau_0$. 
This problem is known to be $\Sigma_1^1$-complete \cite{Harel85}, so if we reduce it to \hyltl 
satisfiability this will establish the desired hardness result. 

The first step in reducing the tiling problem to \hyltl satisfiability is to create a formula that is satisfied only in models that can be interpreted as an $\mathbb{N}\times \mathbb{N}$ grid.
 We do this by using a designated atomic proposition~$x$ that in each trace holds at exactly one time. Furthermore, we require that each trace in a model can be identified by the time at which $x$ holds.
That is, for every $i\in\N$ there is at least one trace where $x$ holds at position $i$, and furthermore we require this trace to be unique.
Hence, we can interpret this trace as $\set{i} \times \nats$.
This will allow us to interpret the entire model as $\N\x\N$.
Note that this means that a trace is a vertical slice of $\nats \times \nats$. We choose to do so, as it simplifies the reduction. 

If we treat each tile in $\mathcal{T}$ as an atomic proposition, it is then easy to create a formula that is satisfied exactly if each point of such a grid-like model contains exactly one tile, the tile colours match and the tile $\tau_0$ occurs infinitely often in $\{0\}\times \mathbb{N}$. The (prenex normal form of the) conjunction of the formula establishing the grid and the formula enforcing a recurring tiling then holds if and only if there is a recurring tiling of the positive quadrant with $\mathcal{T}$.

See \cref{sec:hyltlsat-proofs} for a more detailed proof.
\end{proof}

The problem of whether there is a tiling of $\{(i,j)\in \mathbb{N}^2\mid i\geq j\}$, i.e.\ the part of $\mathbb{N}\times\mathbb{N}$ below the diagonal, such that a designated tile $\tau_0$ occurs on every row, is also $\Sigma_1^1$-complete~\cite{Harel85}.\footnote{The proof in \cite{Harel85} is for the part \emph{above} the diagonal with $\tau_0$ occurring on every column, but that is easily seen to be equivalent.} This problem can similarly be reduced to the satisfiability of \hyltl sentences on sets of ultimately periodic traces, see \cref{sec:hyltlsat-proofs} for details. It follows that \hyltl satisfiability on ultimately periodic traces is also $\Sigma^1_1$-hard. A slightly modified version of the $\Sigma^1_1$-membership proof also applies, so the problem is $\Sigma^1_1$-complete.

\begin{theorem}
\label{thm:hyltlsatup-copmleteness}
  HyperLTL satisfiability restricted to sets of ultimately periodic traces is $\Sigma^1_1$-complete.
\end{theorem}

\section{The \hyltl Quantifier Alternation Hierarchy}
\label{sec:hierarchy}
The number of quantifier alternations in
a formula is a crucial parameter in the complexity of \hyltl model-checking
\cite{FinkbeinerRS15,Rabe16}.
A natural question is then to understand \emph{which} properties
can be expressed with $n$ quantifier alternations, that is, given
a sentence~$\varphi$, determine if there exists an equivalent one with at
most $n$ alternations.
In this section, we show that this problem is in fact exactly as hard as the
\hyltl unsatisfiability problem (which asks whether a \hyltl sentence has no model),
and therefore $\Pi^1_1$-complete. Here, $\Pi_1^1$ is the co-class of $\Sigma_1^1$, i.e.\ it contains the complements of the $\Sigma_1^1$ sets.

Formally, the \hyltl quantifier alternation hierarchy is defined as follows.
Let $\varphi$ be a \hyltl formula. We say that $\varphi$ is a $\Sigma_0$- or
a $\Pi_0$-formula if it is quantifier-free. It is a $\Sigma_n$-formula
if it is of the form $\varphi = \exists \pi_1 \cdots \exists \pi_k.\ \psi$
and $\psi$ is a $\Pi_{n-1}$-formula. It is a $\Pi_n$-formula if
it is of the form $\varphi = \forall \pi_1 \cdots \forall \pi_k.\ \psi$ and
$\psi$ is a $\Sigma_{n-1}$-formula.
We do not require each block of quantifiers to be non-empty, i.e.\ we may have
$k=0$ and $\varphi = \psi$.
By a slight abuse of notation, we also let $\Sigma_n$ denote the set of hyperproperties
definable by a $\Sigma_n$-sentence, that is, the set of all 
$L(\varphi) = \{ T \subseteq (2^\AP)^\omega \mid T \models \varphi\}$
such that $\varphi$ is a $\Sigma_n$-sentence of \hyltl.

\begin{theorem}[{\cite[Corollary 5.6.5]{Rabe16}}]
  The quantifier alternation hierarchy of \hyltl is strict:
  for all $n > 0$, $\Sigma_n \subsetneq \Sigma_{n+1}$.
\end{theorem}

The strictness of the hierarchy also holds if we restrict our attention to
sentences whose models consist of finite sets of traces that end in the suffix~$\emptyset^\omega$, i.e.\ that are essentially finite.

\begin{theorem}\label{thm:strictness-finite}
  For all $n > 0$, there exists a $\Sigma_{n+1}$-sentence $\varphi$ of \hyltl
  that is not equivalent to any $\Sigma_{n}$-sentence, and such that for all
  $T \subseteq {(2^{\AP})}^\omega$, if $T \models \varphi$ then $T$ contains
  finitely many traces and $T \subseteq {(2^{\AP})}^\ast\emptyset^\omega$.
\end{theorem}

This fact 
is a necessary ingredient for our argument that membership
at some fixed level of the alternation hierarchy is $\Pi_1^1$-hard.
It could be derived from a small adaptation of the proof in~\cite{Rabe16},
and we provide an alternative proof
in Appendix~\ref{sec:translations} by exhibiting a connection between the \hyltl quantifier alternation hierarchy and the quantifier alternation hierarchy for first-order logic over finite
words, which is known to be strict~\cite{CohenB71,Thomas82}.

Our goal is to prove the following.

\begin{theorem}\label{thm:Pi11}
  Fix $n > 0$. The problem of deciding whether a \hyltl sentence is equivalent
  to some $\Sigma_n$-sentence is $\Pi^1_1$-complete.
\end{theorem}

The easier part is the upper bound, since a corollary of
\cref{thm:hyltl-sat} is that the problem of deciding whether two
\hyltl formulas are equivalent is $\Pi^1_1$-complete.
The lower bound is proven by reduction from the \hyltl unsatisfiability problem.
The proof relies on Theorem~\ref{thm:strictness-finite}: given a sentence
$\varphi$, we are going to combine $\varphi$ with some $\Sigma_{n+1}$-sentence
$\varphi_{n+1}$ witnessing the strictness of the hierarchy, to construct a sentence $\psi$
such that $\varphi$ is unsatisfiable if and only $\psi$ is equivalent to
a $\Sigma_n$-sentence.
Intuitively, the formula $\psi$ will describe models consisting of the
``disjoint union'' of a model of $\varphi_{n+1}$ and a model of~$\varphi$.
Here ``disjoint'' is to be understood in a strong sense: we split both
the set of traces and the time domain into two parts, used respectively to
encode the models of $\varphi_{n+1}$ and those of $\varphi$.

\begin{figure}
\centering
  \begin{tikzpicture}[yscale=0.6,font=\small]
    \node (03) at (0,3) {$\{a,b\}$};
    \node (13) at (1,3) {$\{a\}$};
    \node (23) at (2,3) {$\{a\}$};
    \node (02) at (0,2) {$\{b\}$};
    \node (12) at (1,2) {$\emptyset$};
    \node (22) at (2,2) {$\{a\}$};
    \begin{pgfonlayer}{background}
      \draw[teal,fill,fill opacity = 0.2,semithick]
      (03.north west) rectangle (22.south east) ;
    \end{pgfonlayer}
    \node[teal] at (-1,2.5) {\large $T_\ell$} ;
    
    \node (31) at (3,1) {$\{a\}$};
    \node (41) at (4,1) {$\{a\}$};
    \node (51) at (5,1) {$\{a,b\}$};
    \node (61) at (6,1) {$\emptyset$};
    \node (71) at (7,1) {$\{a\}$};
    \node (81) at (8,1) {${\cdots}$};
    \node (30) at (3,0) {$\{b\}$};
    \node (40) at (4,0) {$\{a\}$};
    \node (50) at (5,0) {$\{b\}$};
    \node (60) at (6,0) {$\{a,b\}$};
    \node (70) at (7,0) {$\{a\}$};
    \node (80) at (8,0) {${\cdots}\vphantom{\{a\}}$};
    \begin{pgfonlayer}{background}
      \draw[brown,fill, fill opacity = 0.2,semithick]
      (31.north west) rectangle (80.south east) ;
      \node[brown] at (9,0.5) {\large $T_r$} ;
    \end{pgfonlayer}

    \foreach \i in {0,...,2} {
      \foreach \j in {0,1} {
        \node (ij) at (\i,\j) {$\{ \$ \}$} ;
      }
    }
    \foreach \i in {3,...,7} {
      \foreach \j in {2,3} {
        \node (ij) at (\i,\j) {$\{ \$ \}$} ;
      }
    }
    \node at (8,3) {${\cdots}$};
    \node at (8,2) {${\cdots}$};
  \end{tikzpicture}
  \caption{Example of a split set of traces where
  each row represents a trace and $b=3$.
      \label{fig:split}}
\end{figure}

\smallskip

To make this more precise, let us introduce some notations.
We assume a distinguished symbol $\$ \notin \AP$.
We say that a set of traces $T \subseteq {(2^{\AP\cup \{\$\}})}^\omega$ is
\emph{bounded} if there exists $b \in \Nat$ such that
$T \subseteq {(2^{\AP})}^b \cdot \{\$\}^\omega$.
We say that $T$ is \emph{split} if there exist $b \in \Nat$ and $T_1$, $T_2$
such that $T = T_1 \uplus T_2$,
$T_1 \subseteq {(2^{\AP})}^b \cdot \{\$\}^\omega$, and
$T_2 \subseteq \{\$\}^b \cdot {(2^{\AP})}^\omega$.
Note that $b$ is unique here.
Hence, we define the left and right part of $T$ as $\Tl = T_1$ and
$\Tr = \{t \in {(2^{\AP})}^\omega \mid \{\$\}^b \cdot t \in T_2\}$, respectively
(see \cref{fig:split}).

It is easy to combine \hyltl specifications for the left and right part
of a split model into one global formula
(cf.\ Appendix~\ref{sec:hierarchy-proofs}).

\begin{lemma}
    \label{lem:dupsi}
  For all \hyltl sentences~$\varphi_\ell, \varphi_r$, one can construct
  a sentence $\psi$ such that for all split 
  $T \subseteq {(2^{\AP\cup \{\$\}})}^\omega$,
  it holds that   $\Tl \models \varphi_\ell$ and $\Tr \models \varphi_r$ if and only if $T \models \psi$.
\end{lemma}

Conversely, any \hyltl formula that only has split models can be decomposed
into a Boolean combination of formulas that only talk about the left or right
part of the model. This is formalised in the lemma below, proven in
Appendix~\ref{sec:hierarchy-proofs}.

\begin{lemma}\label{lem:split}
  For all \hyltl $\Sigma_n$-sentences $\varphi$
  there exists a finite family $(\varphi^i_\ell, \varphi^i_r)_i$ of $\Sigma_n$-sentences 
  such that for all 
  split $T \subseteq {(2^{\AP\cup \{\$\}})}^\omega$: 
  $T \models \varphi$ if and only if there is an $i$ with
  $\Tl \models \varphi^i_\ell$ and $\Tr \models \varphi^i_r$.
\end{lemma}

We are now ready to prove \cref{thm:Pi11}. %

\begin{proof}[Proof of \cref{thm:Pi11}]
  The upper bound is an easy consequence of \cref{thm:hyltl-sat}:
  Given a \hyltl sentence $\varphi$, we can enumerate all $\Sigma_n$-sentences
  $\psi$, and deciding if $\psi$ is equivalent to $\varphi$
  is in $\Pi_1^1$: $\varphi$ and $\psi$ are equivalent if and only if
  $(\lnot \varphi \land \psi) \lor (\varphi \land \lnot \psi)$ is
  not satisfiable.

  We prove the lower bound by reduction from the unsatisfiability problem for
  \hyltl. So given an \hyltl sentence $\varphi$, we want to
  construct $\psi$ such that $\varphi$ is unsatisfiable if and only if
  $\psi$ is equivalent to a $\Sigma_n$-sentence.
  
  Fix a sentence $\Sigma_{n+1}$-sentence $\varphi_{n+1}$ that is in not
  equivalent to any $\Sigma_n$-sentence, and such that every model of
  $\varphi_{n+1}$ is bounded.
  The existence of such a formula is a consequence of
  \cref{thm:strictness-finite}.
  By \cref{lem:dupsi}, there exists a computable $\psi$ such that for all split models
  $T$, we have $T \models \psi$ if and only if $\Tl \models \varphi_{n+1}$
  and $\Tr \models \varphi$.

  First, it is clear that if $\varphi$ is unsatisfiable, then $\psi$ is
  unsatisfiable as well, and thus equivalent to
  $\exists \pi.\ a_\pi \land \lnot a_\pi$, which is a $\Sigma_n$-sentence
  since $n \ge 1$.
  
  Conversely, suppose towards a contradiction that $\varphi$ is satisfiable
  and that $\psi$ is equivalent to some $\Sigma_n$-sentence.
  Let $(\psi^i_\ell,\psi^i_r)_{i}$ be the finite family of $\Sigma_n$-sentences given
  by Lemma~\ref{lem:split} for $\psi$.
  Fix a model $T_\varphi$ of $\varphi$.
  For a bounded $T$, we let $\overline T$ denote the unique split set of traces such that
  $\overline \Tl = T$ and $\overline \Tr = T_\varphi$.
  For all $T$, we then have $T \models \varphi_{n+1}$ if and only if $T$ is
  bounded and $\overline T \models \psi$.
  Note that the set of bounded models can easily be defined by a
  \hyltl sentence $\phib$ (see \cref{lemma:bounded} in Appendix~\ref{sec:hierarchy-proofs}).
  We then have $T \models \varphi_{n+1}$ if and only if $T \models \phib$ and
  there exists $i$ such that $T \models \psi^i_\ell$ and
  $T_\varphi \models \psi^i_r$.
  So $\varphi_{n+1}$ is equivalent to
  \[
    \phib \land \bigvee\nolimits_{i \text{ with } T_{\varphi} \models \psi^i_r}  \psi^i_\ell \, ,
  \]
  which, since $\Sigma_n$-sentences are closed (up to logical equivalence)
  under conjunction and disjunction, is equivalent to a $\Sigma_n$-sentence.
  This contradicts the definition of $\varphi_{n+1}$.
\end{proof}

\section{HyperCTL$^*$ satisfiability is $\Sigma^2_1$-complete}
\label{sec:ctlsat}
Here, we consider the \hyctlstar satisfiability problem: given a \hyltl sentence, determine whether it has a model $\tsys$ (of arbitrary size). We prove that it is much harder than \hyltl satisfiability.
As a key step of the proof, we also prove that every satisfiable sentence
admits a model of cardinality at most $\cont$ (the cardinality of the continuum), and conversely, we exhibit
a satisfiable sentence whose models are all of cardinality at least~$\cont$.

\begin{theorem}\label{thm:hyctl-sat}
  \hyctlstar satisfiability is $\Sigma^2_1$-complete.
\end{theorem}

On the other hand, \hyctlstar satisfiability restricted to finite transition systems is $\Sigma_1^0$-complete. The upper bound follows from \hyctlstar model checking being decidable~\cite{ClarksonFKMRS14} while the matching lower bound is inherited from \hyltl~\cite{FinkbeinerH16}.

\subparagraph*{Upper bound.}

We begin by proving membership in $\Sigma^2_1$. The first step is to obtain a bound on the size of minimal models of satisfiable \hyctlstar sentences. For this, we use a Skolemisation procedure. This procedure is a a transfinite  generalisation of the proof that all satisfiable \hyltl sentences have a  countable model \cite{FZ17}.

In the following, we use $\omega$ and $\omega_1$ to denote the first infinite and the first uncountable ordinal, respectively, and write $\aleph_0$ and $\aleph_1$ for their cardinality. 

\begin{proposition}\label{prop:hyctl-model}
Each satisfiable HyperCTL$^*$ sentence~$\phi$ has a model of size at most~$\cont$.
\end{proposition}
\begin{proof}[Proof sketch]
Suppose $\phi$ has a model $\tsys$ of arbitrary size, and fix Skolem functions witnessing this satisfaction. We then create a transfinite sequence of transition systems~$\tsys_\alpha$. We start by taking $\tsys_0$ to be any single path from $\tsys$ starting in the initial vertex, and obtain $\tsys_{\alpha+1}$ by adding to $\tsys_\alpha$ all vertices and edges of the paths that are the outputs of the Skolem functions when restricted to inputs from $\tsys_\alpha$. If $\alpha$ is a limit ordinal we take $\tsys_\alpha$ to be the union of all previous transition systems.

This sequence does not necessarily stabilise at $\omega$, since $\tsys_{\omega}$ may contain a path~$\rho$ such that
$\rho(i)$ was introduced in $\tsys_i$. This would result in $\tsys_{\omega}$ containing a path that was not present in any earlier model $\tsys_i$ with $i<\omega$, and therefore we could have $\tsys_{\omega+1}\not = \tsys_{\omega}$.

The sequence does stabilise at $\omega_1$, however. This is because every path $\rho$ contains only countably many vertices, so if every element $\rho(i)$ of $\rho$ is introduced at some countable $\alpha_i$, then there is a countable $\alpha$ such that all of $\rho$ is included in $\tsys_\alpha$. It follows that $\tsys_{\omega_1}$ does not contain any ``new'' paths that were not already in some $\tsys_\alpha$ with $\alpha<\omega_1$, and therefore the Skolem function $f$ does not generate any ``new'' outputs either.

In each step of the construction at most $\cont$ new vertices are added, so $\tsys_{\omega_1}$ contains at most $\cont$ vertices. Furthermore, because $\tsys_{\omega_1}$ is closed under the Skolem functions, the satisfaction of $\phi$ in $\tsys$ implies its satisfaction in $\tsys_{\omega_1}$.
Details can be found in \cref{sec:hyctl-proofs}.
\end{proof}

With the upper bound at hand, we can place \hyctlstar satisfiability in $\Sigma_1^2$, as the existence of a model of size~$\cont$ can be captured by quantification over type~$2$ objects.

\begin{proposition}\label{prop:hyctl-membership}
\hyctlstar satisfiability is in $\Sigma_1^2$.
\end{proposition}
\begin{proof}[Proof sketch]
As in the proof of \cref{thm:hyltl-sat}. Because every \hyctlstar formula is satisfied in a model of size at most $\cont$, these models can be represented by objects of type~$2$. Checking whether a formula is satisfied in a transition system is equivalent to the existence of a winning strategy for Verifier in the induced model checking game. Such a strategy is again a type~$2$ object, which is existentially quantified. Finally, whether it is winning can be expressed by quantification over individual elements and paths, which are objects of types~$0$ and $1$.

Checking the satisfiability of a \hyctlstar formula $\phi$ therefore amounts to existential third-order quantification (to choose a model and a winning strategy) followed by a second-order formula to verify that $\phi$ holds on the model. Hence \hyctlstar satisfiability is in $\Sigma^2_1$.
Details are presented in \cref{sec:hyctl-proofs}.
\end{proof}

\subparagraph*{Lower bound.}

We first describe a satisfiable \hyctlstar sentence $\phiset$ that does not have any model of
cardinality less than $\cont$ (more precisely, the initial vertex must have
uncountably many successors), thus matching the upper bound from
Proposition~\ref{prop:hyctl-model}.
We construct $\phiset$ with one particular model~$\Kset$ in mind,
defined below, though it also admits other models.

The idea is that we want all possible subsets of $A \subseteq \Nat$ to
be represented in $\Kset$ in the form of paths
$\rho_A$ such that $\rho_A(i)$ is labelled by $1$ if $i \in A$, and by $0$ otherwise.
By ensuring that the first vertices of these paths are pairwise distinct, we obtain the desired lower bound on the cardinality.
We express this in \hyctlstar as follows: First, we express that there is a part of the model (labelled by $\fbt$) where every reachable vertex has two successors, one labelled with $0$ and one labelled with $1$, i.e.\ the unravelling of this part contains the full binary tree.
Thus, this part has a path~$\rho_A$ %
as above for every subset~$A$, but their initial vertices are not necessarily distinct.
Hence, we also express that there is another part (labelled by $\pset$) that contains a copy of each path in the $\fbt$-part, and that these paths indeed start at distinct successors of the initial vertex.

\begin{figure}
    \centering
\def\dist{2.0cm}
\begin{tikzpicture}[->,
>=stealth', 
level/.style={sibling distance = 3cm/#1, level distance = \dist}, 
scale=0.7,
transform shape,
grow=left,thick]

\node (init) [state] {\phantom{0}}
child {
    node (t0) [state] {0} 
    child {
        node (t00) [state] {0} 
        child {edge from parent[dashed] }
        child {edge from parent[dashed] }
    }
    child {
        node (t01) [state] {1} 
        child {edge from parent[dashed] }
        child {edge from parent[dashed] }
    }
}
child {
    node (t1) [state] {1} 
    child {
        node (t10) [state] {0} 
        child {edge from parent[dashed] }
        child {edge from parent[dashed] }
    }
    child {
        node (t11) [state] {1} 
        child {edge from parent[dashed] }
        child {edge from parent[dashed] }
    }
}
;

{\color{red}
\node[state,right=6cm of t00] (s00) {0}; 
\node[state,right=1cm of s00] (s000) {0}; 
\node[state,right=1cm of s000] (s0000) {0}; 
\draw (init) -- (s00);
\draw (s00) -- (s000);
\draw (s000) -- (s0000);
\draw[dotted] (s0000) -- +(\dist,0);

\draw[dotted] (init) -- (2.5,1.33);
\draw[-,dotted] (2.5,1.33) -- +(5,0);
\draw[dotted] (init) -- (2.5,1.66);
\draw[-,dotted] (2.5,1.66) -- +(5,0);

\draw[dotted] (init) -- (2.5,0.25);
\draw[-,dotted] (2.5,0.25) -- +(5,0);

\draw[dotted] (init) -- (2.5,-0.25);
\draw[-,dotted] (2.5,-0.25) -- +(5,0);

\draw[dotted] (init) -- (2.5,-1.33);
\draw[-,dotted] (2.5,-1.33) -- +(5,0);
\draw[dotted] (init) -- (2.5,-1.66);
\draw[-,dotted] (2.5,-1.66) -- +(5,0);

\node[state,right=6cm of t01] (s01) {0}; 
\node[state,right=1cm of s01] (s010) {1}; 
\draw (init) -- (s01);
\draw (s01) -- (s010);
\draw[dotted] (s010) -- +(3.5cm,0);

\node[state,right=6cm of t10] (s10) {1}; 
\node[state,right=1cm of s10] (s100) {0}; 
\draw (init) -- (s10);
\draw (s10) -- (s100);
\draw[dotted] (s100) -- +(3.5,0);

\node[state,right=6cm of t11] (s11) {1}; 
\node[state,right=1cm of s11] (s111) {1}; 
\node[state,right=1cm of s111] (s1111) {1}; 
\draw (init) -- (s11);
\draw (s11) -- (s111);
\draw (s111) -- (s1111);
\draw[dotted] (s1111) -- +(\dist,0);

}
\draw[->] (0,1cm) -- (init);

\end{tikzpicture}
    \caption{A depiction of $\Kset$. Vertices in black (on the left including the initial vertex) are labelled by $\fbt$, those in red (on the right, excluding the initial vertex) are labelled by $\pset$.
}%

    \label{fig:phiset}
\end{figure}
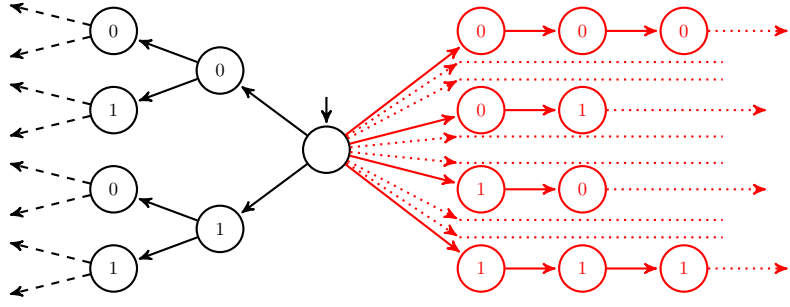

We let $\Kset = (V_\cont,E_\cont,t_\varepsilon,\lambda_\cont)$ (see \cref{fig:phiset}), where
\begin{align*}
  V_\cont & = \{t_u \mid u \in \{0,1\}^\ast\}
      \cup \{s^i_A \mid i \in \Nat \land A \subseteq \Nat\} \\
  \lambda_\cont(t_\varepsilon) & = \{\fbt\} \, \quad
  \lambda_\cont(t_{u \cdot 0}) = \{\fbt,0\} \, \quad
  \lambda_\cont(t_{u \cdot 1}) = \{\fbt,1\} \, \quad
  \lambda_\cont(s^i_A) = \begin{cases}
    \{\pset,0\} & \text{if } i \notin A \\
    \{\pset,1\} & \text{if } i \in A
  \end{cases} \\
  E_\cont
    & = \{(t_u,t_{u0}),(t_u,t_{u1}) \mid u \in \{0,1\}^\ast\} \cup {} 
      \{ (t_\varepsilon, s^0_A) \mid A \subseteq \Nat \} \cup
      \{ (s^i_A,s^{i+1}_A) \mid A \subseteq \Nat, i \in \Nat \} \, .
\end{align*}

\begin{lemma}\label{lem:phiset}
There is a satisfiable \hyctlstar sentence~$\phiset$ that has only models of cardinality at least $\contcard$.
\end{lemma}

\begin{proof}[Proof sketch]

  The formula $\phiset$ states that
  \begin{enumerate}
  \item the label of the initial vertex is $\{\fbt\}$ and the labels of
    non-initial vertices are $\{\fbt,0\}$, $\{\fbt,1\}$, $\{\pset,0\}$ or
    $\{\pset,1\}$;
  \item all $\fbt$-labelled vertices have a successor with label $\{\fbt,0\}$
    and one with label $\{\fbt,1\}$, and no $\pset$-labelled successor;
  \item for every path of $\fbt$-labelled vertices starting at a successor
    of the initial vertex, there is a path of $\pset$-labelled vertices
    (also starting at a successor of the initial vertex) with the same
    $\{0,1\}$ labelling;
  \item any two paths starting in the same $\pset$-labelled vertex have the same sequence of labels.
  \end{enumerate}
  In every model of $\phiset$, for every set~$A$ there is
  a $\pset$-labelled path $\rho_A$ encoding $A$ as above, starting at
  a successor of the initial vertex, and such that $A \neq A'$ implies
  $\rho_A(0) \neq \rho_{A'}(0)$.
  Details can be found in Appendix~\ref{sec:hyctl-proofs}.
\end{proof}

Before moving to the proof that \hyctlstar satisfiability is $\Sigma^2_1$-hard,
we introduce one last auxiliary formula that will be used in the reduction,
showing that addition and multiplication can be defined in \hyctlstar,
and in fact even in \hyltl, as follows:
  Let $\AP = \{\argl,\argr,\res,\add,\mult\}$ and let $\Top$ be the set of all traces $t \in {(2^{\AP})}^\omega$ such that 
  \begin{itemize}
  \item there are unique $n_1,n_2,n_3 \in \nats$ with $\argl \in t(n_1)$,
    $\argr \in t(n_2)$, and $\res \in t(n_3)$, and
  \item either $\add \in t(n)$ for all $n$ and $n_1+n_2 = n_3$,
    or $\mult \in t(n)$ for all $n$ and $n_1 \cdot n_2 = n_3$.
  \end{itemize}
The following lemma is proved in Appendix~\ref{sec:hyctl-proofs}.

\begin{lemma}\label{lem:phiop}

  There is a \hyltl sentence $\phiop$ which has $\Top$ as unique model.
\end{lemma}

To establish $\Sigma_1^2$-hardness, we give an encoding of formulas of
existential third-order arithmetic into \hyctlstar.
As explained in Section~\ref{sec:definitions}, we can (and do for the remainder
of the section) assume that first-order (type 0) variables range over natural
numbers,
second-order (type 1) variables range over sets of natural numbers,
and third-order (type~$2$) variables range over sets of sets of natural numbers.

\begin{lemma}\label{lem:hyctlstar-reduction}
  Suppose $\phi = \exists \xt_1 \ldots \exists \xt_n.\ \psi$, where
  $\xt_1, \ldots, \xt_n$ are third-order variables,
  and $\psi$ is formula of second-order arithmetic.
  One can construct a \hyctlstar formula $\varphi'$ such that
  $(\Nat,0,1, +, \cdot, <,\in)$ is a model of $\varphi$ if and only if
  $\varphi'$ is satisfiable.
\end{lemma}

\begin{proof}  
  The idea of the proof is as follows. We represent sets of natural numbers
  as infinite paths with labels in $\{0,1\}$, so that quantification over
  sets of natural numbers in $\psi$ can be replaced by \hyctlstar path
  quantification.
    First-order quantification is handled in the same way, but using
  paths where exactly one vertex is labelled $1$.
In particular we encode first- and second-order variables~$x$ of $\varphi$ as path variables~$\pix{x}$ of $\varphi'$.
  For this to work, we need to make sure that every possible set
  has a path representative in the transition system (possibly several
  isomorphic ones).
  This is where formula $\phiset$ defined in Lemma \ref{lem:phiset} is used.
  For arithmetical operations, we rely on the formula $\phiop$ from
  Lemma~\ref{lem:phiop}.
  Finally, we associate with every existentially quantified third-order
  variable $\xt_i$ an atomic proposition $\ai i$, so that
  for a second-order variable $\ys$, $\ys \in \xt_i$ is interpreted as the
  atomic proposition $\ai i$ being true on $\pix y$.
  This is all explained in more details below.
  \smallskip
  
  Let $\AP = \{a_1,\ldots,a_n,0,1,\pset,\fbt,\argl,\argr,\res,\mult,\add\}$.
  From the formulas $\phiset$ and $\phiop$ defined in Lemmas~\ref{lem:phiset}
  and \ref{lem:phiop}, it is not difficult to construct a formula
  $\varphi_0$ such that:
  \begin{itemize}
  \item Every transition system constructed by extending $\Kset$ with
    all the traces in $\Top$, added as disjoint paths below
    the initial vertex, and any $\{\ai 1,\ldots,\ai n\}$ labelling,
    is a model of $\varphi_0$.
  \item In any $\tsys = (V,E,v_\initmark, \lambda)$ such that $\tsys \models \varphi$, the following conditions
    are satisfied:
    \begin{enumerate}
    \item For every $A \subseteq \Nat$, there exists a $\pset$-labelled path
      $\rho_A$ starting at a successor of $v_\initmark$ 
      such that $1 \in \lambda(\rho_A(i))$ if and only if $i \in A$,
      and $0 \in \lambda(\rho_A(i))$ if and only if $i \notin A$.
      Moreover, all such paths have the same $\{\ai 1,\ldots,\ai n\}$ labelling;
      this can be expressed by the formula
      \[
        \forall \pi.\ \forall \pi'.\
        \big(\X \G (\pset_\pi \land \pset_{\pi'} \land
        (1_\pi \leftrightarrow 1_{\pi'})\big) \rightarrow
        \bigwedge\nolimits_{\ax \in \{\ai 1, \ldots \ai n\}}
        \G (\ax_\pi \leftrightarrow \ax_{\pi'})
        \, .
      \]
    \item Every path starting at an $\add$- or $\mult$-labelled successor
      of the initial vertex has a label in $\Top$, and vice-versa.
    \end{enumerate}
  \end{itemize}

  We then let $\varphi' = \varphi_0 \land \widehat \psi$, where $\widehat \psi$
  is defined inductively from $\psi$ as follows:
  \begin{itemize}
  \item If $x$ ranges over sets of natural numbers,
    $\exists x. \psi'$ is replaced with
    $\exists \pix x.\ ((\X \pset_{\pix x}) \land \widehat {\psi'})$,
    and if $x$ ranges over natural numbers, with
    $\exists \pix x.\ ((\X \pset_{\pix x}) \land
    \X (0_{\pix x} \U (1_{\pix x} \land \X\G 0_{\pix x})) \land \widehat {\psi'})$.
  \item Predicates $\ys \in \xt_i$ where $\ys$ ranges over sets of natural
    numbers are replaced with $\X (\ai i)_{\pix \ys}$.
  \item Predicates $\xo \in \yt$ where $\xo$ ranges over natural numbers
    and $\yt$ over sets of natural numbers are replaced with
    $\F(1_{\pix \xo} \land 1_{\pix \yt})$.
  \item Predicates $\xo < \yo$ are replaced with
    $\F(1_{\pix \xo} \land \X \F 1_{\pix \yo})$.
  \item Predicates $\xo+\yo = \zo$ are replaced with
    $\exists \pi.\ (\X \add_\pi) \land \F(\argl_\pi \land 1_{\pix \xo}) \land
    \F(\argr_{\pix \xo} \land 1_{\pix \yo}) \land \F(\res_\pi \land 1_{\pix \zo})$,
    and similarly for $\xo \cdot \yo = \zo$.
  \end{itemize}

  If $\psi$ is true under some interpretation of $\xt_1, \ldots, \xt_n$
  as sets of sets of natural numbers,
  then we can construct a model of $\varphi'$ by choosing the
  $\{\ai 1,\ldots,\ai n\}$-labelling of a path $\rho_A = s^0_A s^1_A \ldots$
  in $\Kset$ as follows: $\ai i \in \lambda(\rho_A(0))$ if $A$ is in the
  interpretation of $\xt_i$.
  Conversely, if $\T \models \varphi'$ for some transition system $\T$,
  then for all sets $A \subseteq \Nat$ there is a path $\rho_A$ matching $A$ in $\T$,
  and all such paths have the same $\{\ai 1,\ldots,\ai n\}$-labelling, so we
  can define an interpretation of $\xt_1,\ldots,\xt_n$ by taking
  a set $A$ in the interpretation of $\xt_i$ if and only if $\ai i \in \lambda(\rho_A(0))$. Under this interpretation $\psi$ holds, and thus
  $\varphi$ is true.
\end{proof}

\begin{proposition}
  \hyctlstar satisfiability is $\Sigma^2_1$-hard.
\end{proposition}

\begin{proof}
  Let $N$ be a $\Sigma^2_1$ set, i.e.\
  $N = \set{x \in \nats \mid \exists x_0 \cdots \exists x_k.\ \psi(x, x_0, \ldots, x_k)}$ for some second-order arithmetic formula $\psi$ with existentially
  quantified third-order variables $x_i$.
  For every $n \in \Nat$, we can construct by inlining $n$ in
  $\exists x_0 \cdots \exists x_k.\ \psi(x, x_0, \ldots, x_k)$
  a sentence $\varphi_n$ such that $\varphi_n$ is true  if and only if  $n \in N$.
  Combining this with Lemma~\ref{lem:hyctlstar-reduction}, we obtain a
  computable function that maps any $n \in \Nat$ to a \hyctlstar formula
  $\phi'_n$ such that $n \in N$ if and only if $\phi'_n$ is satisfiable.
\end{proof}

\section{Conclusion}
\label{sec:conc}
In this work, we have settled the complexity of the satisfiability problems for \hyltl and \hyctlstar. In both cases, we significantly increased the lower bounds, i.e.\ from $\Sigma_1^0$ and $\Sigma_1^1$ to  $\Sigma_1^1$ and $\Sigma_1^2$, respectively, and presented the first upper bounds, which are tight in both cases.
Along the way, we also determined the complexity of restricted variants, e.g.\ \hyltl satisfiability restricted to ultimately periodic traces (or, equivalently, to finite traces) is still $\Sigma_1^1$-complete while \hyctlstar satisfiability restricted to finite transition systems is $\Sigma_1^0$-complete.
As a key step in this proof, we showed a tight bound of $\contcard$ on the size of minimal models for satisfiable \hyctlstar sentences.
Finally, we also show that deciding membership in any level of the \hyltl quantifier alternation hierarchy is $\Pi_1^1$-complete.

\subparagraph*{Acknowledgements} We thank Karoliina Lehtinen and Wolfgang Thomas for fruitful discussions. 

\newpage

\bibliography{content/references}

\clearpage
\appendix

\section{Proofs Omitted from Section~\ref{sec:sat}}
\label{sec:hyltlsat-proofs}
\subsection{Proof of Theorem~\ref{thm:hyltl-sat}}

The proof of $\Sigma_1^1$-completeness of the \hyltl satisfiability problem is split into two parts.

\begin{lemma}
\hyltl satisfiability is in $\Sigma_1^1$.	
\end{lemma}

\begin{proof}
Let $\varphi$ be a \hyltl formula, let $\Phi$ denote the set of quantifier-free subformulas of $\varphi$, and let $\Pi$ be a trace assignment whose domain contains the variables of $\varphi$. The expansion of $\varphi$ on $\Pi$ is the function~$\expansion_{\varphi, \Pi} \colon \Phi \times \nats \rightarrow \set{0,1}$ with
\[
\expansion_{\varphi, \Pi}(\psi, j)= \begin{cases}
1 &\text{if $\suffix{\Pi}{j} \models \psi$, and}\\
0 &\text{otherwise.}
\end{cases} 
\]
The expansion is completely characterised by the following consistency conditions:
\begin{itemize}
	\item $\expansion_{\varphi, \Pi} (a_\pi,j)=1$ if and only if $a \in \Pi(\pi)(j)$.
	\item $\expansion_{\varphi, \Pi} (\neg \psi,j)=1$ if and only if  $\expansion_{\varphi, \Pi} (\psi,j)=0$.
	\item $\expansion_{\varphi, \Pi} (\psi_1 \vee \psi_2,j)=1$ if and only if $\expansion_{\varphi, \Pi} (\psi_1,j)=1$ or $\expansion_{\varphi, \Pi} (\psi_2,j)=1$.
	\item $\expansion_{\varphi, \Pi} (\X\psi,j)=1$ if and only if $\expansion_{\varphi, \Pi} (\psi,j+1)=1$.
	\item $\expansion_{\varphi, \Pi} (\psi_1 \U \psi_2,j)=1$ if and only if there is a $j' \ge j$ such that $\expansion_{\varphi, \Pi} (\psi_2,j')=1$ and $\expansion_{\varphi, \Pi} (\psi_2,j'')=1$ for all $j'$ in the range~$j \le j'' < j'$.
\end{itemize}

Every satisfiable \hyltl sentence has a countable model~\cite{FZ17}.
Hence, to prove that the \hyltl satisfiability problem is in $\Sigma_1^1$, we express, for a given \hyltl sentence encoded as a natural number, the existence of the following type~$1$ objects (relying on the fact that there is a bijection between finite sequences over $\nats$ and $\nats$ itself):
\begin{itemize}

	\item A countable set of traces over the propositions of $\varphi$ encoded as a function~$T$ from~$\nats \times \nats$ to $\nats$, mapping trace names and positions to (encodings of) subsets of the set of propositions appearing in $\varphi$.
	
	\item A function~$S$ from $\nats \times \nats^*$ to $\nats$ to be interpreted as Skolem functions for the existentially quantified variables of $\varphi$, i.e.\ we map a variable (identified by a natural number) and a trace assignment of the variables preceding it (encoded as a sequence of natural numbers) to a trace name.

	\item A function~$E$ from $\nats \times \nats \times \nats$ to $\nats$, where, for a fixed~$a \in\nats$ encoding a trace assignment~$\Pi$, the function~$x,y\mapsto E(a, x, y)$ is interpreted as the expansion of $\varphi$ on $\Pi$, i.e.\ $x$ encodes a subformula in $\Phi$ and $y$ is a position.

\end{itemize}
Then, we express the following properties using only type~$0$ quantification: For every trace assignment of the variables in $\varphi$, encoded by $a \in \nats$, if $a$ is consistent with the Skolem function encoded by $S$, then the function~$x,y\mapsto E(a, x, y)$ satisfies the consistency conditions characterizing the expansion, and we have $E(a,x_0, 0) = 1$, where $x_0$ is the encoding of the maximal quantifier-free subformula of $\varphi$.
 We leave the tedious, but standard, details to the industrious reader. 
\end{proof}

\begin{lemma}
\hyltl satisfiability is $\Sigma_1^1$-hard.	
\end{lemma}

\begin{proof}
We reduce the recurring tiling problem to \hyltl satisfiability. 
In this reduction, each $x$-coordinate in the positive quadrant will be represented by a trace, and each $y$-coordinate by a point in time. In order to keep track of which trace represents which $x$-coordinate, we use one designated atomic proposition $x$ that holds on exactly one time point in each trace: $x$ holds at time $i$ if and only if the trace represents $x$-coordinate $i$.

For this purpose, let $\mathcal{C}$ and $\mathcal{T}$ be given, and define the following formulas over $\ap = \set{ x } \cup \mathcal{T}$:

\begin{itemize}
  \item Every trace has exactly one point where $x$ holds:
  \[\varphi_1 = \forall \pi.\ (\neg x_\pi \U (x_\pi\wedge \X\G\neg x_\pi))\]
  \item   For every $i\in\mathbb{N}$, there is a trace with $x$ in the $i$-th position:
  \[\phi_2 = (\exists \pi.\ x_\pi) \wedge (\forall \pi_1.\ \exists \pi_2.\ \F(x_{\pi_1}\wedge \X x_{\pi_2}))\]
  \item
  If two traces represent the same $x$-coordinate, then they contain the same tiles:
  \[\varphi_3 = \forall \pi_1.\ \forall\pi_2.\ (\F(x_{\pi_1}\wedge x_{\pi_2})\rightarrow \G(\bigwedge_{\tau\in \mathcal{T}}(\tau_{\pi_1}\leftrightarrow \tau_{\pi_2})))\]
  \item Every time point in every trace contains exactly one tile:
  \[\varphi_4=\forall \pi.\ \G\bigvee_{\tau\in \mathcal{T}}(\tau_\pi\wedge \bigwedge_{\tau'\in T\setminus \{\tau\}}\neg (\tau')_\pi)\]
  \item Tiles match vertically:
  \[\varphi_5=\forall \pi.\ \G\bigvee_{\tau\in \mathcal{T}}(\tau_\pi\wedge \bigvee\nolimits_{\tau'\in\{\tau'\in \mathcal{T}\mid \tau(\north)=\tau'(\south)\}}\X (\tau')_\pi)\]
  \item Tiles match horizontally:
  \[\varphi_6=\forall \pi_1.\ \forall\pi_2.\ (\F(x_{\pi_1}\wedge \X x_{\pi_2})\rightarrow \G\bigvee_{\tau\in \mathcal{T}}(\tau_{\pi_1}\wedge \bigvee\nolimits_{\tau'\in \{\tau'\in \mathcal{T}\mid \tau(\east)=\tau'(\west)\}}(\tau')_{\pi_2}))\]
  \item Tile~$\tau_0$ occurs infinitely often at $x$-position $0$:
  \[\varphi_7=\exists \pi.\ (x_\pi \wedge \G\F \tau_0)\]
\end{itemize}

Finally, take $\varphi_{\mathcal{T}} = \bigwedge_{1\leq i \leq 7}\varphi_i$. Technically $\varphi_{\mathcal{T}}$ is not a \hyltl formula, since it is not in prenex normal form, but it can be trivially transformed into one. Collectively, subformulas $\varphi_1$--$\varphi_3$ are satisfied in exactly those sets of traces that can be interpreted as $\mathbb{N}\times\mathbb{N}$. Subformulas $\varphi_4$--$\varphi_6$ then hold if and only if the $\mathbb{N}\times\mathbb{N}$ grid is correctly tiled with $\mathcal{T}$. Subformula $\varphi_7$, finally, holds if and only if the tiling uses the tile $\tau_0$ infinitely often at $x$-coordinate $0$. Overall, this means $\varphi_{\mathcal{T}}$ is satisfiable if and only if $\mathcal{T}$ can recurrently tile the positive quadrant.

The $\Sigma^1_1$-hardness of \hyltl satisfiability therefore follows from the $\Sigma_1^1$-hardness of the recurring tiling problem~\cite{Harel85}.
\end{proof}

\subsection{Proof of Theorem~\ref{thm:hyltlsatup-copmleteness}}

Recall that we need to prove that the \hyltl satisfiability problem restricted to ultimately periodic traces is $\Sigma_1^1$-complete.

\begin{proof}[Proof of \cref{thm:hyltlsatup-copmleteness}]
By a reduction of the existence problem for a tiling $\{(i,j)\in \mathbb{N}\times \mathbb{N}\mid i\geq j\}$ with the property that $\tau_0$ occurs at least once on each row.

That problem can be reduced to HyperLTL satisfiability on ultimately periodic traces. The reduction in question is very similar to the one discussed above, with the necessary changes being: (i) every time point beyond $x$ satisfies the special tile ``null'', (ii) horizontal and vertical matching are only checked at or before time point $x$ and (iii) for every $\pi_1$ there is a $\pi_2$ such that $\pi_2$ has designated tile~$\tau_0$ at the time where $\pi_1$ satisfies $x$ (so $\tau_0$ holds at least once in every row).
\end{proof}

\section{Strictness of the Quantifier Alternation Hierarchy}
\label{sec:translations}
The goal of this section is to prove \cref{thm:strictness-finite}, i.e.\ the strictness of the \hyltl quantifier alternation hierarchy, with a witness formula whose models are finite sets of traces in $(2^\AP)^\ast\cdot \emptyset^\omega$.
We rely on the fact that the quantifier alternation hierarchy of first-order
logic over finite words ($\FOw$) is strict \cite{CohenB71,Thomas82}, and embed
it into the \hyltl hierarchy.
The proof is organised as follows. We first define an encoding of finite words
as sets of traces. We then show that every first-order formula can be
translated into an equivalent (modulo encodings) HyperLTL formula with the
same quantifier prefix.
Finally, we show how to translate back HyperLTL formulas into $\FOw$ formulas
with the same quantifier prefix, so that if the HyperLTL alternation quantifier
hierarchy collapsed, then so would the hierarchy for $\FOw$.

\subparagraph{First-Order Logic over Words.}
Let $\AP$ be a finite set of atomic propositions. A finite word over $\AP$
is a finite sequence $w = w(0)w(1) \cdots w(k)$ with $w(i) \in 2^\AP$ for all
$i$. We let $|w|$ denote the \emph{length} of $w$, and
$\pos(w) = \{0,\ldots,|w|-1\}$ the set of \emph{positions} of $w$.
The set of all finite words over $\AP$ is ${(2^\AP)}^\ast$.

Assume a countably infinite set of variables $\varfo$.
The set of $\FOw$ formulas is given by the grammar
\[
  \varphi ::= a(x) \mid x \le y \mid\lnot \varphi
  \mid \varphi \lor \varphi \mid  \exists x.\ \varphi \mid \forall x.\ \varphi \, ,
\]
where $a \in \AP$ and $x,y \in \varfo$.
The set of free variables of $\varphi$ is denoted $\Free(\varphi)$.
A sentence is a formula without free variables.

The semantics is defined as follows, $w \in {(2^\AP)}^\ast$ being a finite word
and $\nu : \Free(\varphi) \to \pos(w)$ an interpretation mapping variables to
positions in $w$:
\begin{itemize}
\item $(w, \nu) \models a(x)$ if $a \in w(\nu(x))$.
\item $(w, \nu) \models x \le y$ if $\nu(x) \le \nu(y)$.
\item $(w, \nu) \models \lnot \varphi$ if $w, \nu \not\models \varphi$.
\item $(w, \nu) \models \varphi \lor \psi$ if $w, \nu \models \varphi$ or
  $(w, \nu) \models \psi$.
\item $(w, \nu) \models \exists x.\ \varphi$ if there exists a position
  $n \in \pos(w)$ such that $(w,\nu[x \mapsto n]) \models \varphi$.
  \item $(w, \nu) \models \forall x.\ \varphi$ if for all positions
  $n \in \pos(w)$: $(w,\nu[x \mapsto n]) \models \varphi$.
\end{itemize}
If $\varphi$ is a sentence, we write
$w \models \varphi$ instead of $(w,\nu) \models \varphi$.

As for \hyltl, a $\FOw$ formula in prenex normal form is a $\Sigma_n$-formula if its quantifier
prefix consists of $n$ alternating blocks of
quantifiers (some of which may be empty), starting with a block of
existential quantifiers.
We let $\Sfo n$ denote the class of languages of finite words definable
by $\Sigma_n$-sentences.

\begin{theorem}[\cite{Thomas82,CohenB71}]\label{thm:strictness-FO}
  The quantifier alternation hierarchy of $\FOw$ is strict: for all
  $n \ge 0$, $\Sfo {n} \subsetneq \Sfo {n+1}$.
\end{theorem}

\subparagraph{Encodings of Words.}
The idea is to encode a word $w \in {(2^\AP)}^\ast$ as a set of traces
$T$ where each trace in $T$ corresponds to a position in $w$;
letters in the word are reflected in the label of the first
position of the corresponding trace in $T$, while the total order $<$ is
encoded using a fresh proposition $o \notin \ap$. More precisely, each trace
has a unique position labelled~$o$, distinct from one trace to another,
and traces are ordered according to the order of appearance of
the proposition $o$.
Note that there are several possible encodings for a same word, and we may fix
a canonical one when needed.
This is defined more formally below.

A \emph{stretch function} is a monotone funtion~$f: \Nat \to \Nat\setminus\set{0}$, i.e.\ it satisfies $0 < f(0) < f(1) < \cdots$.
For all words $w \in {(2^{\AP})}^\ast$ and stretch functions~$f$, we define the set of traces~$\enc w f = \{t_n \mid n \in \pos(w)\} \subseteq (2^{\AP \cup \{o\}})^\ast\emptyset^\omega$
as follows: for all $i \in \Nat$,
\begin{itemize}
\item for all $a \in \AP$,
  $a \in t_n(i)$ if and only if $i = 0$ and $a \in w(n)$
\item $o \in t_n(i)$ if and only if $i = f(n)$.
\end{itemize}

It will be convenient to consider encodings with arbitrarily large spacing
between $o$'s positions.
To this end, for every $N \in \Nat$, we define a particular encoding
\[
  \encn w N = \enc w {n \mapsto N(n+1)} \, .
\]
So in $\encn w N$, two positions with non-empty labels are at distance at
least $N$ from one another.

Given $T = \enc w f$ and a trace assignment
$\Pi \colon \var \rightarrow T$,
we let $\Tn T N = \encn w N$, and
$\nun \Pi N \colon \var \rightarrow \Tn T N$ the trace assignment defined
by shifting the $o$ position in each $\Pi(\pi)$ accordingly, i.e.,
\begin{itemize}
	\item $o \in \Pi^{(N)}(\pi) (N(i+1))$ if and only if $o \in \Pi(\pi)(f(i))$ and
	\item for all $a \in \ap$: $a \in \Pi^{N}(\pi)(0)$ if and only if $a \in \Pi(\pi)(0)$.
\end{itemize}

\subparagraph{From FO to HyperLTL.}
We associate with every $\FOw$ formula $\varphi $ in prenex normal form
a \hyltl formula $\encPhi \varphi$ over $\ap \cup \{o\}$ by replacing
in $\varphi$:
\begin{itemize}
\item $a(x)$ with $a_x$, and
\item $x \le y$ with $\F(o_x \land \F o_y)$. 
\end{itemize}
In particular, $\encPhi \varphi$ has the same quantifier prefix as $\varphi$, which means that we treat variables of $\varphi$ as trace variables of $\encPhi{\varphi}$.

\begin{lemma}\label{lem:FO-to-hyltl}
  For every $ \FOw$ sentence $\varphi $ in prenex normal form,
  $\varphi$ is equivalent to $\encPhi \varphi$ in the following sense:
  for all $w \in {(2^{\AP})}^\ast$ and all stretch functions~$f$,
  \[
    w \models \varphi \quad\text{if and only if}\quad
    \enc w f \models \encPhi \varphi \, .
  \]
\end{lemma}

In particular, note that the evaluation of $\encPhi \varphi$ on $\enc w f$ does
not depend on $f$. We call such a formula \emph{stretch-invariant}:
a \hyltl sentence  $\varphi$ is \emph{stretch-invariant} if for all finite words $w$ and all stretch functions~$f$ and $g$,
\[
  \enc w f \models \varphi \quad\text{ if and only if }\quad
  \enc w g \models \varphi \, .
\]

\begin{lemma}\label{lem:stretchinv}
  For all $\varphi \in \FOw$, $\encPhi \varphi$ is stretch-invariant.
\end{lemma}

\subparagraph{Going Back From HyperLTL to FO.}
Let $\simpleHyperLTL$ denote the fragment of $\HyperLTL$ consisting of
all formulas $\encPhi \varphi$, where $\varphi$ is a $\FOw$ formula in
prenex normal form.
Equivalently, $\psi \in \simpleHyperLTL$ if it is a HyperLTL formula of the
form $\psi = Q_1x_1 \cdots Q_kx_k.\ \psi_0$, where $\psi_0$ is a Boolean
combination of formulas of the form  $a_x$ or $\F(o_x \land \F o_y)$.

Let us prove that every \hyltl sentence is equivalent, over
sets of traces of the form~$\enc w f$, to a sentence in $\simpleHyperLTL$
with the same quantifier prefix.
This means that if a \hyltl sentence $\encPhi \varphi$ is equivalent to a \hyltl
sentence with a smaller number of quantifier alternations, then it is also
equivalent over all word encodings to one of the form $\encPhi \psi$, which
in turns implies that the $\FOw$ sentences $\varphi$ and $\psi$ are equivalent.

The \emph{temporal depth} of a quantifier-free formula in \hyltl
is defined inductively as
\begin{itemize}\label{defdepth}
\item 	$\depth(a_\pi) = 0$,
\item  $\depth(\lnot \varphi) = \depth(\varphi)$,
\item $\depth(\varphi \lor \psi) = \max(\depth(\varphi),\depth(\psi))$,
\item $\depth(\X \varphi) = 1 + \depth(\varphi)$, and
\item $\depth(\varphi \U \psi) = 1 + \max(\depth(\varphi,\psi))$.
\end{itemize}
For a general \hyltl formula $\varphi = Q_1 \pi_1 \cdots Q_k \pi_k.\ \psi$,
we let $\depth(\varphi) = \depth(\psi)$.

\begin{lemma}\label{lem:QFsimpl}
  Let $\psi$ be a quantifier-free formula of $ \HyperLTL$.
  Let $N = \mathit{depth}(\psi)+1$.
  There exists a quantifier-free formula $\simpl \psi \in \simpleHyperLTL$
  such that for all $T = \enc w f$ and trace assignments $\Pi$,
  \[
    (\Tn T N, \nun \Pi N) \models \psi \quad\text{if and only if}\quad (T,\Pi) \models \simpl \psi
     \, .
  \]
\end{lemma}

\begin{proof} 
  Assume that $\Free(\psi) = \{\pi_1, \ldots, \pi_k\}$ is the set of free variables of $\psi$.
  Note that the value of $(\Tn T N, \nun \Pi N) \models \psi$ depends
  only on the traces $\nun \Pi N(\pi_1),\ldots,\nun \Pi N(\pi_k)$.
  We see the tuple $(\nun \Pi N(\pi_1),\ldots,\nun \Pi N(\pi_k))$ as a single
  trace $w_{T,\Pi,N}$ over the set of propositions
  $\AP' = \{a_{\pi} \mid a \in \AP\cup\set{o} \land \pi \in \Free(\psi)\}$,
  and $\psi$ as an LTL formula over $\AP'$.
  
  We are going to show that the evaluation of $\psi$ over words $w_{T,\Pi,N}$ is
  entirely determined by the ordering of $o_{\pi_1}, \ldots, o_{\pi_n}$ in
  $w_{T,\Pi,N}$ and the label of $w_{T,\Pi,N}(0)$, which we can both
  describe using a formula in $\simpleHyperLTL$.
  The intuition is that non-empty labels in $w_{T,\Pi,N}$ are
  at distance at least $N$ from one another, and  a temporal formula of depth
  less than $N$ cannot distinguish between $w_{T,\Pi,N}$ and other words
  with the same sequence of non-empty labels and sufficient spacing between
  them.
  More generally, the following can be easily proved via
  Ehrenfeucht-Fra\"iss\'e games:

  \begin{claim}\label{claim:EF}
    Let $m,n \ge 0$, $(a_i)_{i \in \Nat}$ be a sequence of letters in
    $2^{\AP'}$, and
    \[
      w_1, w_2 \in
      \emptyset^m a_0 \emptyset^{n}\emptyset^\ast
      a_1 \emptyset^{n}\emptyset^\ast a_2 \emptyset^{n}\emptyset^\ast \cdots
    \]
    Then for all LTL formulas $\varphi$ such that $\depth(\varphi) \le n$,
    $w_1 \models \varphi$ if and only if $w_2 \models \varphi$.
  \end{claim}

  Here we are interested in words of a particular shape.
  Let $L_N$ be the set of infinite words $w \in {(2^{\AP'})}^\omega$ such that:
  \begin{itemize}
  \item For all $\pi \in \{\pi_1,\ldots,\pi_k\}$, there is a unique
    $i \in \Nat$ such that $o_{\pi} \in w(i)$. Moreover, $i \ge N$.
  \item If $o_{\pi} \in w(i)$ and $o_{\pi'} \in w(i') $,
    then $|i-i'| \ge N$ or $i = i'$.
  \item If $a_{\pi} \in w(i)$ for some $a \in \AP$ and
    $\pi \in \{\pi_1,\ldots,\pi_k\}$, then $i = 0$.
  \end{itemize}
  Notice that $w_{T,\Pi,N} \in L_N$ for all $T$ and all $\Pi$.
  
  For $w_1, w_2 \in L_N$, we write $w_1 \sim w_2$ if $w_1$ and $w_2$ differ
  only in the spacing between non-empty positions, that is, if there are
  $\ell \le k$ and $a_0, \ldots, a_\ell \in 2^{\AP'}$ such that
  $w_1,w_2 \in a_0 \emptyset^\ast a_1 \emptyset^\ast \cdots a_\ell \emptyset^\omega$.    
  Notice that $\sim$ is of finite index.
  Moreover, we can distinguish between its equivalence classes using formulas
  defined as follows. For all
  $A \subseteq \{a_{\pi} \mid a \in \AP \land \pi \in \{\pi_1,\ldots,\pi_k\}\}$
  and all total preorders $\preceq$ over $\{\pi_1, \ldots, \pi_k\}$,\footnote{
    i.e.\ $\preceq$ is transitive and for all
    $\pi,\pi' \in \{\pi_1, \ldots, \pi_k\}$, $\pi \preceq \pi'$
    or $\pi' \preceq \pi$ (or both)} we let
  \[
    \varphi_{A,\preceq} =
    \bigwedge_{a \in A} a \land \bigwedge_{a \notin A} \lnot a
    \land \bigwedge_{\pi_i \preceq \pi_j} \F(o_{\pi_i} \land \F o_{\pi_j})
    \, .
  \]
  Note that every word $w \in L_N$ satisfies exactly one formula
  $\varphi_{A,\preceq}$, and that all words in an equivalence class
  satisfy the same one.
  We denote by $L_{A,\preceq}$ the equivalence class of $L_N/{\sim}$
  consisting of words satisfying $\varphi_{A,\preceq}$.
  So we have $L_N = \biguplus L_{A,\preceq}$.

  Since $\psi$ is of depth less than $N$, by \cref{claim:EF}
  (with $n = N-1$ and $m = 0$), for all $w_1 \sim w_2$ we have
  $w_1 \models \psi$ if and only if $w_2 \models \psi$.
  Now, define $\simpl \psi$ as the disjunction of all $\varphi_{A,\preceq}$
  such that $\psi$ is satisfied by elements in the class $L_{A,\preceq}$.
  Then $\simpl \psi \in \simpleHyperLTL$, and
  \[
    \text{for all } w \in L_N, \quad w \models \simpl \psi
    \text{ if and only if } w \models \psi \, .
  \]
  In particular, for every $T$ and every $\Pi$, we have
  $(\Tn T N,\nun \Pi N) \models \simpl \psi$ if and only if 
  $(\Tn T N,\nun \Pi N) \models \psi$.
  Since the preorder between propositions $o_\pi$ and the label of the initial
  position are the same in $(\Tn T N,\nun \Pi N)$ and $(T,\Pi)$,
  we also have $(T,\Pi) \models \simpl \psi$ if and only if
  $(\Tn T N,\nun \Pi N) \models \simpl \psi$.
  Therefore,
  \[
(\Tn T N, \nun \Pi N) \models \psi \quad\text{if and only if}\quad     (T,\Pi) \models \simpl \psi
     \, . \qedhere
  \]
\end{proof}

For a quantified \hyltl sentence $\varphi = Q_1\pi_1 \cdots Q_k\pi_k.\ \psi$, we let
$\simpl \varphi = Q_1\pi_1 \ldots Q_k\pi_k.\ \simpl \psi$, where $\simpl \psi$ is
the formula obtained through \cref{lem:QFsimpl}.

\begin{lemma}\label{lem:Qsimpl}
  For all \hyltl formulas~$\varphi $,
  for all $T = \enc w f$ and trace assignments~$\Pi$,
  \[
 (\Tn T N, \nun \Pi N) \models \varphi    \text{ if and only if }
    (T, \Pi) \models \simpl \varphi \, ,
  \]
  where $N = \mathit{depth}(\varphi)+1$.
\end{lemma}
\begin{proof}
  We prove the result by induction. We have
  \begin{align*}
    (T, \Pi) \models \exists \pi.\ \simpl \psi
    & \quad\Leftrightarrow\quad
      \exists t \in T \text{ such that }
      (T,\Pi[\pi \mapsto t]) \models \simpl \psi \\
    & \quad\Leftrightarrow\quad
      \exists t \in T \text{ such that }
      (\Tn T N, \nun {(\Pi[\pi \mapsto t])} N)
      \models \psi && \text{(IH)} \\
    & \quad\Leftrightarrow\quad
        \exists t \in \Tn T N \text{ such that }
      (\Tn T N, \nun {\Pi} N[\pi \mapsto t])
      \models \psi \\
    & \quad\Leftrightarrow\quad
        (\Tn T N, \nun \Pi N) \models \exists \pi.\ \psi \, ,
  \end{align*}
  and similarly,
  \begin{align*}
    (T, \Pi) \models \forall \pi.\ \simpl \psi
    & \quad\Leftrightarrow\quad
      \forall t \in T, \text{we have }
      (T,\Pi[\pi \mapsto t]) \models \simpl \psi \\
    & \quad\Leftrightarrow\quad
      \forall t \in T, \text{we have }
      (\Tn T N, \nun {(\Pi[\pi \mapsto t])} N)
      \models \psi && \text{(IH)} \\
    & \quad\Leftrightarrow\quad
        \forall t \in \Tn T N, \text{we have }
      (\Tn T N, \nun \Pi N [\pi \mapsto t])
      \models \psi \\
    & \quad\Leftrightarrow\quad
        (\Tn T N, \nun \Pi N) \models \forall \pi.\ \psi \, . && \qedhere
  \end{align*}
\end{proof}

As a corollary, we obtain:

\begin{lemma}\label{lem:eqsimpl}
  For all stretch-invariant \hyltl sentences $\varphi$ and
  for all $T = \enc w f$,
  \[
    T \models \varphi \quad\text{if and only if}\quad
    T \models \simpl \varphi \, .
  \]
\end{lemma}

\begin{proof}
  By definition of $\varphi$ being stretch-invariant, we have
  $T \models \varphi$ if and only if
  $\Tn T N \models \varphi$, which by \cref{lem:Qsimpl} is
  equivalent to $T \models \simpl \varphi$.
\end{proof}

We are now ready to prove the strictness of the \hyltl quantifier alternation
hierarchy.

\begin{proof}[Proof of \cref{thm:strictness-finite}.]
  Suppose towards a contradiction that the hierarchy collapses at
  level~$n > 0$, i.e.\ every \hyltl $\Sigma_{n+1}$-sentence is equivalent to some
  $\Sigma_n$-sentence.
  Let us show that the $\FOw$ quantifier alternation hierarchy also
  collapses at level $n$, a contradiction with \cref{thm:strictness-FO}.

  Fix a $\Sigma_{n+1}$-sentence $\varphi$ of $\FOw$.
  The \hyltl sentence $\encPhi \varphi$ has the same quantifier prefix as
  $\varphi$, i.e.\ is also a $\Sigma_{n+1}$-sentence.
  Due to the assumed hierarchy collapse, there exists a  \hyltl
  $\Sigma_n$-sentence $\psi$ that is equivalent to $\encPhi \varphi$,
  and is stretch-invariant by \cref{lem:stretchinv}.
  Then the \hyltl sentence $\simpl \psi$ defined above is also a $\Sigma_n$-sentence.
  Moreover, since $\simpl \psi \in \simpleHyperLTL$, there exists
  a $\FOw$ sentence $\varphi'$ such that $\simpl \psi = \encPhi {\varphi'}$,
which has the same quantifier prefix as $\simpl \psi$, i.e.\
  $\varphi'$ is a $\Sigma_n$-sentence of $\FOw$.
  For all words $w \in (2^\AP)^\ast$, we now have
  \begin{align*}
    w \models \varphi
    & \quad\text{if and only if}\quad \enc w f \models \encPhi \varphi
    && \text{(\cref{lem:FO-to-hyltl})} \\
    & \quad\text{if and only if}\quad \enc w f \models \psi
    && \text{(assumption)} \\
    & \quad\text{if and only if}\quad \enc w f \models \simpl \psi
    && (\text{\cref{lem:eqsimpl,lem:stretchinv}}) \\
    & \quad\text{if and only if}\quad \enc w f \models \encPhi {\varphi'}
    && (\text{definition}) \\
    & \quad\text{if and only if}\quad w \models \varphi'
    && \text{(\cref{lem:FO-to-hyltl})}
  \end{align*}
  for an arbitrary encoding function $f$.
  Therefore, $\Sfo {n+1} = \Sfo n$,
  yielding the desired contradiction.

  This proves not only that for all $n > 0$, there is a \hyltl $\Sigma_{n+1}$-sentence
  that is not equivalent to any $\Sigma_n$-sentence,
  but also that there is one of the form $\encPhi \varphi$.
  Now, the proof still goes through if we replace $\encPhi \varphi$ by
  any formula equivalent to $\encPhi \varphi$ over all $\enc{w}{f}$,
  and in particular if we replace $\encPhi \varphi$ by
  $\encPhi \varphi \land \psi$, where the sentence
  \[
    \psi = \exists \pi.\ \forall \pi'.\
    (\F \G \emptyset_\pi) \land
    \G( \G \emptyset_\pi \rightarrow  \G \emptyset_{\pi'}) \, ,
    \qquad
    \text{with }
    \emptyset_\pi = \bigwedge_{a \in \AP} \lnot a_\pi \, ,
  \]
  selects models that contain finitely many traces, all in
  $(2^\AP)^\ast \cdot \emptyset^\omega$.
  Indeed, all $\enc w f$ satisfy~$\psi$.
  Notice that $\psi$ is a $\Sigma_2$-sentence, and since $n+1 \ge 2$,
  (the prenex normal form of) $\encPhi \varphi \land \psi$ is still a $\Sign {n+1}$-sentence.
\end{proof}

\section{Proofs Omitted from Section~\ref{sec:hierarchy}}
\label{sec:hierarchy-proofs}

\subsection{Proof of Lemma~\ref{lem:dupsi}}

We need to prove  that for all \hyltl sentences~$\varphi_\ell, \varphi_r$, one can construct
  a sentence $\psi$ such that for all split 
  $T \subseteq {(2^{\AP\cup \{\$\}})}^\omega$
  it holds that $\Tl \models \varphi_\ell$ and $\Tr \models \varphi_r$ if and only if
  $T \models \psi$.

\begin{proof}[Proof of \cref{lem:dupsi}]
  Let $\widehat {\varphi_r}$ denote the formula obtained from $\varphi_r$
  by replacing:
  \begin{itemize}
  \item every existential quantification $\exists \pi.\ \varphi$ with
    $\exists \pi.\ ((\F\G \lnot \$_\pi) \land \varphi)$;
  \item every universal quantification $\forall \pi.\ \varphi$ with
    $\forall \pi.\ ((\F\G \lnot \$_\pi) \rightarrow \varphi)$;
  \item the quantifier-free part $\varphi$ of $\varphi_r$ with
    $\$_\pi \U (\lnot \$_\pi \land \varphi)$,
    where $\pi$ is some free variable in $\varphi$.
  \end{itemize}
 Here, the first two replacements restrict quantification to traces in the right part while the last one requires the formula to hold at the first position of the right part. 
  We define $\widehat {\varphi_\ell}$ by similarly relativizing quantifications
  in $\varphi_\ell$.
  The formula $\widehat {\varphi_\ell} \land \widehat {\varphi_r}$ can then
  be put back into prenex normal form to define $\psi$.
\end{proof}

\subsection{Proof of Lemma~\ref{lem:split}}

We need to prove that for all \hyltl $\Sigma_n$-sentences $\varphi$,
  there exists a finite family of $\Sigma_n$-sentences
  $(\varphi^i_\ell, \varphi^i_r)_{i}$
  such that for all 
  split $T \subseteq {(2^{\AP\cup \{\$\}})}^\omega$: 
  $T \models \varphi$ if and only if there exists
  $i$ such that
  $\Tl \models \varphi^i_\ell$ and $\Tr \models \varphi^i_r$.
  To prove this result by induction, we need to strengthen the statement to make it dual and allow for formulas with free variables.
  We let $\Free(\varphi)$ denote the set of free variables of a formula
  $\varphi$. 
  We prove the following result, which implies \cref{lem:split}.

\begin{lemma}\label{lem:splitgen}
  For all HyperLTL $\Sigma_n$-formulas (resp.\ $\Pin n$-formulas) $\varphi$,
  there exists a finite family of $\Sigma_n$-formulas
  (resp.\ $\Pi_n$-formulas)
  $(\varphi^i_\ell, \varphi^i_r)_{i}$
  such that for all $i$,
  $\Free(\varphi) = \Free(\varphi^i_\ell) \uplus \Free(\varphi^i_r)$,
  and for all split $T$ and $\Pi$: $(T, \Pi) \models \varphi$
  if and only if there exists $i$ such that
  \begin{itemize}
  \item For all $\pi \in \Free(\varphi)$, $\Pi(\pi) \in \Tl$ if and only if
    $\pi \in \Free(\varphi^i_{\ell})$
    (and thus $\Pi(\pi) \in T \setminus \Tl$ if and only if
    $\pi \in \Free(\varphi^i_{r})$).
  \item $(\Tl, \Pi) \models \varphi^i_\ell$;
  \item $(\Tr, \Pi') \models \varphi^i_r$,
    where $\Pi'$ maps every $\pi \in \Free(\varphi^i_r)$ to the trace in
    $\Tr$ corresponding to $\Pi(\pi)$ in $T$
    (i.e.\ $\Pi(\pi) = \set{\$}^b \cdot \Pi'(\pi)$ for some $b$).
  \end{itemize}
\end{lemma}

\begin{proof}
  To simplify, we can assume that the partition of the free variables of
  $\varphi$ into a left and right part is fixed, i.e.\
  we take $\Vl \subseteq \Free(\varphi)$
  and $\Vr = \Free(\varphi) \setminus \Vl$, and we restrict our attention
  to split $T$ and $\Pi$ such that $\Pi(\Vl) \subseteq \Tl$
  and $\Pi(\Vr) \subseteq T \setminus \Tl$.
  The formulas $(\varphi^i_\ell, \varphi^i_r)_{i}$ we are looking
  for should then be such that $\Free(\varphi^i_\ell) = \Vl$ and
  $\Free(\varphi^i_r) = \Vr$.
  If we can define sets of formulas $(\varphi^i_\ell, \varphi^i_r)$ for each
  choice of $\Vl, \Vr$, then the general case is solved by taking the union
  of all of those.
  So we focus on a fixed $\Vl,\Vr$, and prove the result by induction on
  the quantifier depth of $\varphi$.
  
  \subparagraph{Base case.}
  If $\varphi$ is quantifier-free, then it can be seen as an LTL formula
  over the set of propositions
  $\{a_\pi, \$_\pi \mid \pi \in \Free(\varphi), a \in \AP\}$,
  and any split model of $\varphi$ consistent with $\Vl,\Vr$
  can be seen as a word in $\Sigma_\ell^\ast \cdot \Sigma_r^\omega$,
  where
  \begin{align*}
    \Sigmal & = \big\{ \alpha \cup \{\$_\pi \mid \pi \in \Vr\} \mid 
    \alpha \subseteq {\{a_\pi \mid \pi \in \Vl \land a \in \AP\}} \big\} \text{ and }\\
    \Sigmar & = \big\{ \alpha \cup \{\$_\pi \mid \pi \in \Vl\} \mid 
    \alpha \subseteq {\{a_\pi \mid \pi \in \Vr \land a \in \AP\}} \big\} \, .
  \end{align*}
  Note in particular that $\Sigmal \cap \Sigmar = \emptyset$.
  We can thus conclude by applying the following standard result of
  formal language theory:
  
  \begin{lemma}\label{lem:splitLTL}
    Let $L \subseteq \Sigma_1^\ast \cdot \Sigma_2^\omega$, where
    $\Sigma_1 \cap \Sigma_2 = \emptyset$.
    If $L = L(\varphi)$ for some LTL formula $\varphi$, then there exists a finite family~$(\varphi^i_1, \varphi^i_2)_{i}$  of LTL formulas 
    such that $L = \bigcup_{1 \le i \le k} L(\varphi^i_1) \cdot L(\varphi^i_2)$
    and for all $i$, $L(\varphi^i_1) \subseteq \Sigma_1^\ast$
    and $L(\varphi^i_2) \subseteq \Sigma_2^\omega$.
  \end{lemma}

  \begin{proof}
    A language is definable in LTL if and only if it is accepted by some
    counter-free automaton~\cite{DG08,Thomas81}.
    Let $\mathcal{A}$ be a counter-free automaton for $L$.
    For every state $q$ in $\mathcal{A}$, let
    \begin{align*}
      L^q_1 & = \{w \in \Sigma_1^\ast \mid q_0 \xrightarrow{w} q
              \text{ for some initial state $q_0$}\} \text{ and }\\
      L^q_2 & = \{w \in \Sigma_2^\omega \mid \text{there is an accepting run
              on $w$ starting from $q$}\} \, .
    \end{align*}
    We have $L = \bigcup_q L^q_1\cdot  L^q_2$.
    Moreover, $L^q_1$ and $L^q_2$ are still recognisable by counter-free
    automata, and therefore LTL definable.
  \end{proof}
  
  \subparagraph{Case $\varphi = \exists \pi.\ \psi$.}
  Let $(\psi^i_{\ell,1},\psi^i_{r,1})$ and $(\psi^i_{\ell,2},\psi^i_{r,2})$
  be the formulas constructed respectively for
  $(\psi,\Vl \cup \{\pi\},\Vr)$ and $(\psi,\Vl,\Vr \cup \{\pi\})$.
  We take the union of all $(\exists \pi.\ \psi^i_{\ell,1},\psi^i_{r,1})$
  and $(\psi^i_{\ell,2},\exists \pi.\ \psi^i_{r,2})$.

  \subparagraph{Case $\varphi = \forall \pi.\ \psi$.}
  Let $(\xi^i_{\ell},\xi^i_{r})_{1 \le i \le k}$ be the formulas obtained for
  $\exists \pi.\ \lnot \psi$.
  We have ${(T,\Pi) \models \varphi}$ if and only if for all $i$,
  $(\Tl,\Pi) \not\models \xi^i_{\ell}$ or $(\Tr,\Pi') \not\models \xi^i_r$;
  or, equivalently, if there exists $h: \{1,\ldots,k\} \to \{\ell,r\}$ such
  that $(\Tl,\Pi) \models \bigwedge_{h(i) = \ell} \lnot \xi^i_{\ell}$ and
  $(\Tr,\Pi') \models \bigwedge_{h(i) = r} \lnot \xi^i_{r}$.
  Take the family $(\varphi^h_{\ell},\varphi^h_{r})_h$, where
  $\varphi^h_{\ell} = \bigwedge_{h(i) = \ell} \lnot \xi^i_{\ell}$ and
  $\varphi^h_{r} = \bigwedge_{h(i) = r} \lnot \xi^i_{r}$.
  Since $\varphi = \forall \pi.\ \psi$ is a $\Pi_n$-formula,
  the formula $\exists \pi.\ \lnot \psi$ and by induction all
  $\xi^i_{\ell}$ and $\xi^i_{r}$ are $\Sigma_n$-formulas.
  Then all $\lnot \xi^i_{r}$ are $\Pi_n$-formulas, and since
  $\Pi_n$-formulas are closed under conjunction
  (up to formula equivalence), all $\varphi^h_{\ell}$ and $\varphi^h_{r}$
  are $\Pi_n$-formulas as well.
\end{proof}

\subsection{Definition of \boldmath$\phib$}

\begin{lemma}\label{lemma:bounded}
  There exists a \hyltl sentence $\phib$ %
  such that for all $T \subseteq {(2^{\AP\cup \{\$\}})}^\omega$, we have
  $T \models \phib$ if and only if $T$ is bounded.
\end{lemma}

\begin{proof}

  We let
    \begin{align*}
    \phib =
    \forall \pi.\ \forall \pi'.\
    & (\lnot \$_\pi \U \G \$_\pi) \land
      \bigwedge_{a \in \AP} \G(\lnot (a_\pi \land \$_\pi)) \land
      \F \left(
      \lnot \$_\pi \land \lnot \$_{\pi'} \land \X \$_\pi \land \X \$_{\pi'}
      \right) \, .
    \end{align*}
    The conjunct $(\lnot \$_\pi \U \G \$_\pi) \land
    \bigwedge_{a \in \AP} \G(\lnot (a_\pi \land \$_\pi))$
    ensures that every trace is in ${(2^{\AP})}^\ast \cdot \{\$\}^\omega$,
    while
    $\F \left(
      \lnot \$_\pi \land \lnot \$_{\pi'} \land \X \$_\pi \land \X \$_{\pi'}
    \right)$
    ensures that the $\$$'s in any two traces $\pi$ and $\pi'$ start at the same
    position.
\end{proof}

\section{Proofs Omitted from Section~\ref{sec:ctlsat}}
\label{sec:hyctl-proofs}
\subsection{Proof of Proposition~\ref{prop:hyctl-model}}

The proof of Proposition~\ref{prop:hyctl-model} uses a Skolem function to create a model. Before giving this proof, we should therefore first introduce Skolem functions for \hyctlstar.

Let $\phi$ be a \hyctlstar formula. A quantifier in $\phi$ occurs \emph{with polarity 0} if it occurs inside the scope of an even number of negations, and \emph{with polarity 1} if it occurs inside the scope of an odd number of negations. We then say that a quantifier \emph{occurs existentially} if it is an existential quantifier with polarity 0, or a universal quantifier with polarity 1. Otherwise the quantifier \emph{occurs universally}. A Skolem function will map choices for the universally occurring quantifiers to choices for the existentially occurring quantifiers.

For reasons of ease of notation, it is convenient to consider a single Skolem function for all existentially occurring quantifiers in a \hyctlstar formula $\phi$, so the output of the function is an $l$-tuple of paths, where $l$ is the number of existentially occurring quantifiers in~$\phi$. The input consists of a $k$-tuple of paths, where $k$ is the number of universally occurring quantifiers in $\phi$, plus an $l$-tuple of integers. The reason for these integers is that we need to keep track of the time point in which the existentially occurring quantifiers are invoked.

Consider, for example, a \hyctlstar formula of the form $\forall \pi_1.\ \G \exists \pi_2.\ \psi$. This formula states that for every path $\pi_1$, and for every future point $\pi_1(i)$ on that path, there is some $\pi_2$ starting in $\pi_1(i)$ satisfying $\psi$. So the choice of $\pi_2$ depends not only on $\pi_1$, but also on $i$. For each existentially occurring quantifier, we need one integer to represent this time point at which it is invoked. A \hyctlstar Skolem function for a formula $\phi$ on a transition system $\tsys$ is therefore a function $f:\mathit{paths}(\tsys)^k\times \mathbb{N}^l\rightarrow \mathit{paths}(\tsys)^l$, where $\mathit{paths}(\tsys)$ is the set of paths over $\tsys$, $k$ is the number of universally occurring quantifiers in $\phi$ and $l$ is the number of existentially occurring quantifiers.

Now, we are able to prove that every satisfiable \hyctlstar formula has a  model of size~$\cont$.

\begin{proof}[Proof of Proposition~\ref{prop:hyctl-model}]
If $\phi$ is satisfiable, let $\tsys$ be one of its models, and let $f$ be a Skolem function witnessing the satisfaction of $\phi$ on $\tsys$. We create a sequence of transition systems $\tsys_\alpha$ as follows.
\begin{itemize}
	\item $\tsys_0$ contains a single, arbitrarily chosen, path of $\tsys$ starting in the initial vertex.
	\item $\tsys_{\alpha+1}$ contains exactly those vertices and edges from $\tsys$ that are (i) part of $\tsys_{\alpha}$ or (ii) among the outputs of the Skolem function $f$ when restricted to input paths from $\tsys_\alpha$.
	\item if $\alpha$ is a limit ordinal, then $\tsys_\alpha= \bigcup_{\alpha'<\alpha} \tsys_{\alpha'}$.
\end{itemize}
Note that if $\alpha$ is a limit ordinal then $\tsys_\alpha$ may contain paths~$\rho(0)\rho(1)\rho(2)\cdots$ that are not included in any $\tsys_{\alpha'}$ with $\alpha' <\alpha$, as long as each every finite prefix $\rho(0) \cdots \rho(i)$ is included in some $\alpha'_i<\alpha$.

First, we show that this procedure reaches a fixed point at $\alpha = \omega_1$. Suppose towards a contradiction that $\tsys_{\omega_1+1}\not =\tsys_{\omega_1}$. Then there are $\vec{\rho}=(\rho_1,\ldots,\rho_k)\in \mathit{paths}(\tsys_{\omega_1})^k$ and $\vec{n}\in \mathbb{N}^l$ such that $f(\vec{\rho},\vec{n})\not \in \mathit{paths}(\tsys_{\omega_1})^l$. Then for every $i\in\mathbb{N}$ and every $1\leq j \leq k$, there is an ordinal $\alpha_{i,j}<\omega_1$ such that the finite prefix $\rho_j(0)\cdots \rho_j(i)$ is contained in $\tsys_{\alpha_{i,j}}$. The set $\{\alpha_{i,j}\mid i\in \mathbb{N}, 1\leq j \leq k\}$ is countable, and because $\alpha_{i,j}<\omega_1$ each $\alpha_{i,j}$ is also countable. A countable union of countable sets is itself countable, so $\sup \{\alpha_{i,j}\mid i\in \mathbb{N}, 1 \leq j \leq k\}=\bigcup_{i\in\mathbb{N}}\bigcup_{1\leq j \leq k}\alpha_{i,j}=\beta< \omega_1$. 

But then the $\vec{\rho}$ are all contained in $\tsys_\beta$, and therefore $f(\vec{\rho},\vec{n})\in \mathit{paths}(\tsys_{\beta+1})^l$. But $\beta+1< \omega_1$, so this contradicts the assumption that $f(\vec{\rho},\vec{n})\not \in \mathit{paths}(\tsys_{\omega_1})^l$. From this contradiction we obtain $\tsys_{\omega_1+1}=\tsys_{\omega_1}$, so we have reached a fixed point. Furthermore, because $\tsys_{\omega_1}$ is contained in $\tsys$ and closed under the Skolem function and $\tsys$ satisfies $\phi$, we obtain that $\tsys_{\omega_1}$ also satisfies $\phi$.

Left to do, then, is to bound the size of $\tsys_{\omega_1}$, by bounding the number of vertices that get added at each step in its construction. We show by induction that $\size{\tsys_\alpha}\leq \cont$ for every $\alpha$. As base case, we have $\size{\tsys_{0}}
\le \aleph_0$, since it consists of a single path. 

Consider then $\size{\tsys_{\alpha+1}}$. For each possible input to $f$, there are at most $l$ new paths, and therefore at most $\size{\nats \times l}$ new vertices in $\tsys_{\alpha+1}$. Further, there are $\size{\mathit{paths}(\tsys_\alpha)}^k\times \size{\mathbb{N}}^l$ such inputs. By the induction hypothesis, $\size{\tsys_\alpha}\leq \cont$, which implies that $\size{\mathit{paths}(\tsys_\alpha)}\leq \cont$. As such, the number of added vertices in each step is limited to $\cont^k\times \aleph_0^l \times \aleph_0\times l= \cont$. So $\size{\tsys_{\alpha+1}}\leq \size{\tsys_\alpha}+\cont=\cont$.

If $\alpha$ is a limit ordinal, $\tsys_\alpha$ is a union of at most $\aleph_1$ sets, each of which has, by the induction hypothesis, a size of at most $\cont$. Hence $\size{\tsys_\alpha}\leq \aleph_1\times \cont = \cont$.
\end{proof}

\subsection{Proof of Proposition~\ref{prop:hyctl-membership}}

Recall that we need to prove that \hyctlstar satisfiability is in $\Sigma_1^2$.

\begin{proof}[Proof of Proposition~\ref{prop:hyctl-membership}]
We encode the existence of a winning strategy for Verifier in the \hyctlstar model checking game~$\mcgame(\tsys, \varphi)$ induced by a transition system~$\tsys$ and an \hyctlstar formula~$\varphi$. 
This game is played between Verifier and Falsifier, one of them aiming to prove that $\tsys \models \varphi$ and the other aiming to prove $\tsys \not\models \varphi$. It is played in a graph whose positions correspond to subformulas which they want to check (and suitable path assignments of the free variables): each vertex (say, representing a subformula $\psi$) belongs to one of the players who has to pick a successor, which represents a subformula of $\psi$. A play ends at an atomic proposition, at which point the winner can be determined.

Formally, a vertex of the game is of the form~$(\Pi, \psi,b)$ where $\Pi$ is a path assignment, $\psi$ is a subformula of $\varphi$, and $b \in \set{0,1}$ is a flag used to count the number of negations encountered along the play; the initial vertex is $(\Pi_\emptyset, \varphi,0)$.
Furthermore, for until-subformulas~$\psi$, we need auxiliary vertices of the form~$(\Pi,\psi, b, j)$ with $j \in \nats$.
The vertices~$v$ of Verifier are 
\begin{itemize}
	\item of the form $(\Pi, \psi, 0)$ with $\psi = \psi_1 \vee \psi_2$, $\psi = \psi_1 \U \psi_2$, or $\psi = \exists \pi.\ \psi'$, 
	\item of the form $(\Pi, \forall \pi.\ \psi', 1)$, or
	\item of the form $(\Pi, \psi_1 \U \psi_2, 1, j)$.
\end{itemize}
The moves of the game are defined as follows:
\begin{itemize}
	
	\item A vertex~$(\Pi, a_\pi, b)$ is terminal. It is winning for Verifier if $b = 0$ and $a \in \lambda(\Pi(\pi)(0))$ or if $b=1$ and $a \notin \lambda(\Pi(\pi)(0))$, where $\lambda$ is the labelling function of $\tsys$.
	
	\item A vertex~$(\Pi, \neg \psi, b)$ has a unique successor~$(\Pi,\psi, b+1 \bmod 2)$.
	
	\item A vertex~$(\Pi, \psi_1\vee \psi_2, b)$ has two successors of the form~$(\Pi, \psi_i, b)$ for $i \in \set{1,2}$.
	
	\item A vertex~$(\Pi, \X \psi, b)$ has a unique successor~$(\suffix{\Pi}{1}, \psi, b)$.
	
	\item A vertex~$(\Pi, \psi_1\U \psi_2, b)$ has a successor~$(\Pi, \psi_1\U \psi_2, b,j)$ for every $j \in \nats$.
	
	\item A vertex~$(\Pi, \psi_1\U \psi_2, b,j)$ has the  successor~$(\suffix{\Pi}{j}, \psi_2, b)$ as well as successors~$(\suffix{\Pi}{j'}, \psi_1, b)$ for every $0 \le j' <j$.
	
	\item A vertex~$(\Pi, \exists \pi.\ \psi, b)$ has successors~$(\Pi[\pi\mapsto \rho],\psi, b)$ for every path~$\rho$ of $\tsys$ starting in $\last(\Pi)$.
	
	\item A vertex~$(\Pi, \forall \pi.\ \psi, b)$ has successors~$(\Pi[\pi\mapsto \rho],\psi, b)$ for every path~$\rho$ of $\tsys$ starting in $\last(\Pi)$.
\end{itemize}

A play of the model checking game is a finite path through the graph, starting at the initial vertex and ending at a terminal vertex. It is winning for Verifier if the terminal vertex is winning for her. 
Note that the length of a play is bounded by~$2d$, where $d$ is the depth\footnote{Depth is defined similarly to  temporal depth (see Page~\pageref{defdepth}), but takes every Boolean connective and temporal operator into account, not just the temporal ones.} of $\varphi$, as the formula is simplified during each move. 

A strategy~$\sigma$ for Verifier is a function mapping each of her vertices~$v$  to some successor of~$v$. 
A play~$v_0 \cdots v_k$ is consistent with $\sigma$, if $v_{k'+1} = \sigma(v_{k'})$ for every $0 \le k' < k$ such that $v_{k'}$ is a vertex of Verifier.
A straightforward induction shows that Verifier has a winning strategy for $\mcgame(\tsys,\varphi)$ if and only if $\tsys\models\varphi$.
 
Recall that every satisfiable \hyctlstar sentence has a model of cardinality~$\contcard$~(\cref{prop:hyctl-model}). 
Thus, to place \hyctlstar satisfiability in $\Sigma_1^2$, we express, for a given natural number encoding a \hyctlstar formula~$\varphi$, the existence of the following type~$2$ objects (using suitable encodings):
\begin{itemize}
	\item A transition system~$\tsys$ of cardinality~$\contcard$. 
	\item A function~$\sigma$ from $V$ to $V$, where $V$ is the set of vertices of $\mcgame(\tsys,\varphi)$. Note that a single vertex of $V$ is a type~$1$ object. 
\end{itemize}
Then, we express that $\sigma$ is a strategy for Verifier, which is easily expressible using quantification over type~$1$ objects. 
Thus, it remains to express that $\sigma$ is winning by stating that every play (a sequence of type~$1$ objects of bounded length) that is consistent with $\sigma$ ends in a terminal vertex that is winning for Verifier. 
Again, we leave the tedious, but standard details to the reader.
 \end{proof}

\subsection{Proof of Lemma~\ref{lem:phiset}}
We need to prove that there is a satisfiable \hyctlstar sentence that has only uncountable models.

 \begin{proof}[Proof of Lemma~\ref{lem:phiset}]
  The formula $\phiset$ is defined as the conjunction of the formulas below:
  \begin{enumerate}
  \item The label of the initial vertex is $\{\fbt\}$ and the labels of non-initial vertices are $\set{\fbt, 0}$, $\set{\fbt, 1}$, $\set{\pset, 0}$, or $\set{\pset, 1}$:
    \[
      \forall \pi.\ ( \fbt_\pi \land \lnot 0_\pi \land \lnot 1_\pi \land
      \lnot \pset_\pi )\wedge \X \G \big(
      (\pset_\pi \leftrightarrow \lnot \fbt_\pi) \land 
      (0_\pi \leftrightarrow \lnot 1_\pi)      \big)
    \]

  \item All $\fbt$-labelled vertices have a successor with label~$\set{\fbt,0}$ and one with label~$\set{\fbt,1}$, and no $\pset$-labelled successor:
    \[
      \forall \pi.\ \G\big(\fbt_\pi \rightarrow
      (\exists \pi_0.\ \X (\fbt_{\pi_0} \land 0_{\pi_0}))
      \land (\exists \pi_1.\ \X (\fbt_{\pi_1} \land  1_{\pi_1})) \land
      (\forall \pi'.\ \X \fbt_{\pi'})
      \big)
    \]
  \item For every path of $\fbt$-labelled vertices starting at a successor
    of the initial vertex, there is a path of $\pset$-labelled vertices
    (also starting at a successor of the initial vertex) with the same
    $\{0,1\}$ labelling:
    \[
      \forall \pi.\  \big( (\X \fbt_\pi) \rightarrow
      \exists \pi'.\ \X( \pset_{\pi'} \land \G(0_\pi \leftrightarrow 0_{\pi'}))
      \big)
    \]
  \item Any two paths starting in the same $\pset$-labelled vertex have the same sequence of labels:
    \[
      \forall \pi.\ \G \big(\pset_\pi \rightarrow
      \forall \pi'.\ \G(0_\pi \leftrightarrow 0_{\pi'})
      \big) \, .
    \]
  \end{enumerate}

  It is easy to check that $\Kset \models \phiset$.
  Note however that it is not the only model of $\phiset$: for instance,
  some paths may be duplicated, or merged after some steps if their label
  sequences share a common suffix.
  So, consider an arbitrary transition system $\tsys = (V,E,v_\initmark, \lambda)$
  such that $\tsys \models \phiset$.
  By condition 2, for every set $A \subseteq \Nat$,  there is a path
  $\rho_A$ starting at a successor of $v_\initmark$
  such that $\lambda(\rho_A(i)) = \{\fbt,1\}$ if $i \in A$ and
  $\lambda(\rho_A(i)) = \{\fbt,0\}$ if $i \notin A$.
  Condition 3 implies that there is also a $\pset$-labelled path $\rho'_A$
  such that $\rho'_A$ starts at a successor of $v_\initmark$, and has the same
  $\{0,1\}$ labelling as $\rho_A$.
  Finally, by condition 4, if $A \neq B$ then $\rho_A'(0) \neq \rho_{B}'(0)$.
 \end{proof}

\subsection{Proof of Lemma~\ref{lem:phiop}}

We prove that the set of traces $\Top$ can be defined in \hyltl.
 
 \begin{proof}[Proof of Lemma~\ref{lem:phiop}]
  The first condition can be expressed as follows in \hyltl:
  \[
    \forall \pi.\ \Big( \G (\mult_\pi \land \lnot \add_\pi) \lor
    \G (\add_\pi \land \lnot \mult_\pi \Big)
    \land \bigwedge_{a \in \{\argl,\argr,\res\}}
    \lnot a \U (a \land \X \G \lnot a) \, .
  \]
  The second condition can be expressed using inductive definitions
  of addition and multiplication:
  \begin{enumerate}
  \item Addition and multiplication are commutative:
    \[
      \forall \pi.\ \exists \pi'.\
      (\add_\pi \leftrightarrow \add_{\pi'}) \land
      \F(\argl_\pi \land \argr_{\pi'}) \land
      \F(\argr_\pi \land \argl_{\pi'}) \land
      \F(\res_\pi \land \res_{\pi'}) \, .
    \]
  \item $0+0 = 0$, and if $i+j = k$, then $i+ (j+1) = (k+1)$:
    \begin{align*}
      & \exists \pi.\ \add_\pi \land \argl_\pi \land \argr_\pi \land \res_\pi \\
      & \forall \pi.\ \exists \pi'.\ \add_\pi \rightarrow {} \\
      & \quad \add_{\pi'} \land \F(\argl_\pi \land \argl_{\pi'}) \land
        \F(\argr_\pi \land \X \argr_{\pi'}) \land
        \F(\res_\pi \land \X \res_{\pi'})
    \end{align*}
    Together with commutativity, this ensures that all valid traces for addition
    are present.
  \item $i+j = 0$ if and only if ($i = 0$ and $j = 0$), and for $j > 0$,
    $i+j = k$ implies $i + (j-1) = (k-1)$:
    \begin{align*}
      & \forall \pi.\ \add_\pi \rightarrow
        (\res_\pi \leftrightarrow (\argl_\pi \land \argr_\pi)) \\
       & \forall \pi.\ \exists \pi'.\
         (\add_\pi \land \lnot \argr_\pi) \rightarrow {}\\
      & \quad \add_{\pi'} \land \F(\argl_\pi \land \argl_{\pi'}) \land
        \F(\X \argr_\pi \land \argr_{\pi'}) \land
        \F(\X \res_\pi \land \res_{\pi'}) \, .
    \end{align*}
    Together with commutativity, this ensures that \emph{only} valid traces
    for addition are present.
  \item $0\cdot 0 = 0$, and if $i \cdot j = k$, then $i \cdot (j+1) = i+ k$:
    \begin{align*}
      & \exists \pi.\ \mult_\pi \land \argl_\pi \land \argr_\pi \land \res_\pi\\
      & \forall \pi.\ \exists \pi'.\ \exists \pi''.\ \mult_\pi \rightarrow {}\\
      & \quad \mult_{\pi'} \land \add_{\pi''} \land
        \F(\argl_\pi \land \argl_{\pi'} \land \argl_{\pi''}) \land
        \F(\argr_\pi \land \X \argr_{\pi'}) \land {} \\
      & \quad
        \F(\argr_{\pi''} \land \res_\pi) \land
        \F(\res_{\pi''} \land \res_{\pi'}) \, .
    \end{align*}
    Together with commutativity, this ensures that all valid traces for
    multiplication are present.
  \item $i \cdot j = 0$ if and only if ($i = 0$ or $j = 0$),
    and if $j > 0$ and $i \cdot j = k$, then there exists $k'$ such that
    $i \cdot (j-1) = k'$ and $i+k' = k$:
    \begin{align*}
      & \forall \pi.\ \mult_\pi \rightarrow
        (\res_\pi \leftrightarrow (\argl_\pi \lor \argr_\pi)) \\
      & \forall \pi.\ \exists \pi'.\ \exists \pi''.\
         (\mult_\pi \land \lnot \argr_\pi) \rightarrow {}\\
      & \quad \mult_{\pi'} \land \add_{\pi''} \land
        \F(\argl_{\pi} \land \argl_{\pi'} \land \argl_{\pi''}) \land
        \F(\X \argr_{\pi} \land \argr_{\pi'}) \land {} \\
      & \quad \F(\res_{\pi'} \land \argr_{\pi''}) \land
        \F(\res_{\pi''} \land \res_{\pi}) \, .
    \end{align*}
    Together with commutativity, this ensures that \emph{only} valid traces
    for multiplication are present. \qedhere
  \end{enumerate}
\end{proof}

\appendix
\end{document}